\documentclass{amsart}
\usepackage{amsmath,amsfonts,amsbsy,amsgen,amscd,mathrsfs,amssymb,amsthm,mathtools}
\usepackage[colorlinks=true,citecolor=blue,linkcolor=blue]{hyperref}

\usepackage{geometry}
\numberwithin{equation}{section}
\newtheorem{theorem}{Theorem}[section]
\newtheorem{corollary}[theorem]{Corollary}

\newtheorem{conjecture}[theorem]{Conjecture}
\newtheorem{lemma}[theorem]{Lemma}
\newtheorem{proposition}[theorem]{Proposition}
\theoremstyle{definition}
\newtheorem{assumption}[theorem]{Assumption}
\newtheorem{definition}[theorem]{Definition}
\newtheorem{example}[theorem]{Example}
\newtheorem{notation}[theorem]{Notation}
\newtheorem{problem}[theorem]{Problem}

\newtheorem{remark}[theorem]{Remark}

\makeatletter\renewenvironment{proof}[1][\proofname] {\par\pushQED{\qed}\normalfont\topsep6\p@\@plus6\p@\relax\trivlist\item[\hskip\labelsep\bfseries#1\@addpunct{.}]\ignorespaces}{\popQED\endtrivlist}

\newcommand\al{\alpha}
\newcommand\be{\beta}
\newcommand\dd{\mathrm d}
\newcommand\De{\Delta}
\newcommand\de{\delta}
\newcommand\deq{\stackrel{\mathrm d}{=}}
\newcommand\eps{\varepsilon}
\newcommand\Ga{\Gamma}
\newcommand\ga{\gamma}
\newcommand\ka{\kappa}
\newcommand\La{\Lambda}
\newcommand\la{\lambda}

\newcommand\Si{\Sigma}
\newcommand\si{\sigma}

\renewcommand\d{~\mathrm d}
\renewcommand\phi{\varphi}
\renewcommand\rho{\varrho}

\newcommand\bs{\boldsymbol}
\newcommand\mbb{\mathbb}
\newcommand\mbf{\mathbf}
\newcommand\mc{\mathcal}
\newcommand\mf{\mathfrak}
\newcommand\mr{\mathrm}
\newcommand\ms{\mathscr}

\begin{document}

\title[Rigidity of Random Schr\"odinger Operators via Feynman-Kac Formulas]
{Spectral rigidity of random Schr\"odinger operators via Feynman-Kac formulas}
\author{Pierre Yves Gaudreau Lamarre}
\address{Princeton University, Princeton, NJ\space\space08544, USA}
\email{plamarre@princeton.edu}
\author{Promit Ghosal}
\address{Columbia University, New York City, NY\space\space10027, USA}
\email{pg2475@columbia.edu}
\author{Yuchen Liao}
\address{University of Michigan, Ann Arbor, MI\space\space48109, USA}
\email{yuchliao@umich.edu}
\thanks{P. Y. Gaudreau Lamarre was partially funded by a Gordon Y. S. Wu Fellowship.}
\maketitle

\begin{abstract}
We develop a technique for proving number rigidity (in the sense of Ghosh-Peres \cite{GP17})
of the spectrum of general random Schr{\"o}dinger operators (RSOs).
Our method makes use of Feynman-Kac formulas to estimate the variance
of exponential linear statistics of the spectrum in terms of self-intersection local times.

Inspired by recent results concerning Feynman-Kac formulas
for RSOs with multiplicative noise \cite{GaudreauLamarre,GaudreauLamarreShkolnikov,GorinShkolnikov}
by Gorin, Shkolnikov and the first-named author,
we use this method to prove number rigidity for a class of one-dimensional continuous RSOs of the form $-\frac12\De+V+\xi$,
where $V$ is a deterministic potential and $\xi$ is a stationary Gaussian noise.
Our results require only very mild assumptions on the domain on which the operator is defined,
the boundary conditions on that domain, the regularity of the potential $V$, and the singularity of the noise $\xi$.
\end{abstract}

\section{Introduction}

Let $I\subset\mbb R$ be an open interval (possibly unbounded), and let
$V:I\to \mathbb{R}$ be a deterministic potential.
Let $\xi:I\to\mbb R$ be a centered stationary Gaussian process with a covariance of the form
$\mbf E[\xi(x)\xi(y)]=\ga(x-y)$, where $\ga$ is an even function or Schwartz
distribution. (We refer to Section \ref{sec:Noise} for a formal definition.)
In this paper, we investigate the number rigidity of the eigenvalue point processes
of random Schr\"odinger operators (RSOs) of the form
\begin{align}\label{eq:Schr}
\hat{\mc H}_I  := -\tfrac{1}{2}\Delta + V+ \xi,
\end{align}
where $\hat{\mc H}_I$ acts on a subset of
functions $f:I\to\mbb R$ that satisfy some fixed boundary conditions
(if $I$ has a boundary).

\subsection{Random Schr\"odinger Operators}

The spectral theory of RSOs arises naturally in multiple problems in mathematical physics;
we refer to \cite{CarLac} for a general introduction to the subject. Looking more specifically
at one-dimensional continuous operators, RSOs of the form \eqref{eq:Schr} have found
applications in the study of random matrices and interacting particle systems, as well as stochastic partial
differential equations (SPDEs).

Indeed,
if $\xi=\xi_\be$ is a white noise with variance $1/\be$ for some $\be>0$
(see Example \ref{Example: White}) and $V(x)=x/2$,
then the spectrum of $\hat{\mc H}_{(0,\infty)}$
captures the asymptotic edge fluctuations of a large class of {\it $\be$-Ensembles}
\cite{BloemendalVirag,KrishnapurRiderVirag,RamirezRiderVirag} due to its connection
with the {\it Airy-$\be$ process} (see Section \ref{Section: Optimality}).
In another direction, the study of the solutions of SPDEs of the form
\begin{align}\label{eq:DiffEq}
\partial_tu=\tfrac12\Delta u - Vu - \xi u=-\hat{\mc H}_I u
\end{align}
is intimately connected to the spectral theory of $\hat{\mc H}_I $.
More specifically, the {\it localization} of $\hat{\mc H}_I $'s eigenfunctions
is expected to shed light on the geometry of {\it intermittent peaks} in \eqref{eq:DiffEq}
(e.g., \cite[Sections 2.2.3--2.2.4]{KonigBook} and references therein). We refer to
\cite{CarmonaMolchanov,Chen14,Smooth2,Smooth1} for a few examples
of papers where such ideas have been implemented when $\xi$ is a smooth, white, fractional,
or otherwise singular noise (see Examples \ref{Example: White}--\ref{Example: Bounded} for definitions of such noises).

\subsection{Spatial Conditioning and Number Rigidity}
\label{Sec: Spacial Conditioning}

Point processes are well-studied objects in probability \cite{DV08,Ka17},
due to their applications in many disciplines (e.g., \cite{ApplicationsPtProcess}).
One of the simplest point processes is the Poisson process, which is such that
the number of points in disjoint sets are independent. In contrast, for point processes
with strong correlations, the notion of {\it spatial conditioning} (i.e., the distribution of points
inside a bounded set conditional on the point configuration outside the set) is of interest. Pioneering work on this subject
includes the Dobrushin-Lanford-Ruelle (DLR) formalism (e.g.,
\cite[Sections 1.4-2.4]{Dereudre}).

In this paper, we are interested in a form of spatial conditioning
known as \emph{number rigidity} \cite{GP17}. A point process is said to be number rigid if for every
bounded set $A$, the configuration of points outside of $A$ determines the number
of points inside of $A$. We refer to \cite{AM80,HS13} for examples of early work on this kind of property.
In their seminal paper \cite{GP17} (see also \cite{G15}), Ghosh and Peres introduced (among other things) the notion
of number rigidity, and studied its occurrence in two classical point processes. 
Since then, number rigidity has been shown to have many interesting
applications in the theory of point processes (e.g., \cite{BPreprint,BQ17,BQ17_2,G15,SG16,PeresSly}),
and has developed into an active field of research.
We refer to \cite{BBNY16,BBNY17,Buf18,BDQ18,Chatterjee2019,GS19,GhoshKrishnapur,GP17,Olshan11}
for other notions related to number rigidity,
such as {\it higher order/linear rigidity}, {\it hyperuniformity}, {\it sub-extensivity}, {\it quasi-invariance/symmetry}, and {\it tolerance}.

A significant portion of the number rigidity literature is concerned with
uncovering sufficient conditions for certain point processes to be number rigid.
Though several methods have been found to prove number
rigidity (such as DLR equations \cite{DHLM}),
in most cases one proceeds by controlling the variance of carefully chosen sequences
of linear statistics (e.g.,  \cite{Buf16,BNQ18,G15,Ghosh17,GP17,RN18}). We recall that, given a point process $\mf X$ on a domain $\Si$,
the linear statistic associated with a test function $f:\Si\to\mbb R$ is defined as $\sum_{x\in\mf X}f(x).$
Following \cite[Theorem 6.1]{GP17},\footnote{We remark that, although \cite{GP17} pioneered the techniques of controlling the variance of linear the statistics for showing rigidity, this scheme is actually dated back to the works of Kolmogorov \cite{Kol41a,Kol41b} where he derived a sufficient condition for the \emph{linear rigidity} of any stationary sequence. We refer to \cite{BDQ18} and the references therein for more details on linear rigidity.}
if one concocts a sequence
$(f_n)_{n\in\mbb N}$ of test functions such that $f_n\to1$ uniformly on some set $A\subset\Si$
and the variance of $\sum_{x\in\mf X}f_n(x)$ vanishes, then it follows that the number of points
in $A$ is determined by $\mf X$'s configuration outside of $A$ (see Proposition \ref{Proposition: Variance Criterion}).

\subsection{Outline of Results and Method of Proof}
\label{sec:Outline}

To the best of our knowledge, the only RSO whose spectrum is
known to be number rigid is the operator
\[\hat{\mc H}_{(0,\infty)}=-\tfrac12\De+\tfrac x2+\xi_2\]
with a Dirichlet boundary condition at zero,
where $\xi_2$ is a white noise with variance $1/2$.
The proof of this \cite{Buf16} relies on the fact that the eigenvalues
of this operator generate the Airy-$2$ process, which is a determinantal
point process (see \eqref{Eq:Airy Kernel}).
In this context, our main motivation in this paper is to provide a unified framework
to study the number rigidity of the eigenvalues of general RSOs. As a first step in this
direction, we develop a new method of proving number rigidity for RSOs
by controlling the variance of exponential linear statistics
using Feynman-Kac formulas. Informally, our main result is as follows (we point to
Theorems \ref{Theorem: Main} and \ref{Theorem: Main Examples} for precise statements).

\begin{theorem}[Informal Statement]
\label{theorem: informal}
Suppose that $\hat{\mc H}_I$ acts on either the full space $I=\mbb R$,
the half-line $I=(0,\infty)$, or the bounded interval $I=(0,b)$, under some
general boundary conditions in the latter two cases (Assumption \ref{Assumption: Cases}).
Assume that the noise $\xi$ and the deterministic potential
$V$ satisfy mild technical conditions (Assumptions \ref{Assumption: Noise} and \ref{Assumption: Potential}).

On the one hand, when $I$ is unbounded, $\hat{\mc H}_I$'s spectrum is number rigid if
$V$ has sufficient growth at infinity (i.e., \eqref{Equation: General Growth Bounds}
and \eqref{eq:Vwhite}--\eqref{eq:VBound}).
On the other hand, if $I=(0,b)$, then $\hat{\mc H}_{(0,b)}$'s spectrum is always number rigid.
\end{theorem}

Thus, one of the main advantages of the method developed in this paper
is that it applies under very general assumptions on the noise $\xi$,
the domain $I$, the boundary conditions on $I$, and
the regularity of the deterministic potential $V$.
However, in cases where the domain $I$ is unbounded, our method comes at the cost of
growth assumptions on $V$.

\begin{remark}
It is worth noting that our main result does {\it not} imply rigidity of the Airy-$\beta$ process for any $\be>0$,
since our growth condition in the case of white noise requires $V$ to be superlinear (see \eqref{eq:Vwhite}).
In fact, we prove that it is not possible to establish the rigidity of the Airy-$2$ process by using exponential
linear functionals (see Proposition \ref{Proposition: Optimality 2}). This suggests (at least for white noise) that,
while our growth conditions are not necessary for rigidity, they are the optimal conditions that can be obtained with
our semigroup method; see Section \ref{Section: Q of Optimality} for more details.
\end{remark}

The key steps in the proof of our main result are as follows.

\noindent{\bf (i)} We state general conditions (see
Assumptions \ref{Assumption: Noise}, \ref{Assumption: Cases},
and \ref{Assumption: Potential}; and Proposition \ref{prop: F-K})
under which exponential functionals $\mr e^{-t x}$ ($t>0$) of the spectrum of
$\hat{\mc H}_I$ admit a random Feynman-Kac representation.
This follows from a combination of
classical semigroup theory and the
work on Feynman-Kac formulas for RSOs with irregular Gaussian noise
\cite{GaudreauLamarre,GaudreauLamarreShkolnikov,GorinShkolnikov}
pioneered by Gorin and Shkolnikov.

\noindent{\bf (ii)} The Feynman-Kac formulas in {\bf(i)}
give an explicit representation of $\hat{\mc H}_I $'s semigroup in terms of elementary stochastic processes.
This allows to reformulate the vanishing of the variance of exponential linear statistics
in terms of a corresponding limit for the self-intersection local time of Brownian bridges on
$\mathbb{R}$, or reflected Brownian bridges on the half-line or bounded intervals (see
\eqref{Equation: Lp Local Time Scaling} and Theorem \ref{Theorem: Variance Bound}).

\noindent{\bf (iii)}
The main tool we use to control the Brownian bridge self-intersection local time consists of large deviations results
for the self-intersection local time of unconditioned Brownian motion on $\mathbb{R}$.
The latter has been studied extensively; we refer to \cite[Chapter~4]{Chen} and references therein for details.
To bridge the gap between the results on the self-intersection local time of Brownian bridges and the unconditioned Brownian motion,
we make use of couplings between reflected Brownian motions on different domains, and the absolute continuity of the midpoint of bridge processes with respect to
their unconditioned versions.

\noindent{\bf (iv)}
By combining {\bf(i)}--{\bf(iii)}, we obtain our main result (Theorem \ref{Theorem: Main}),
which consists of general sufficient conditions
(see \eqref{Equation: d Exponent} and \eqref{Equation: General Growth Bounds})
for the number rigidity of $\hat{\mc H}_I$'s spectrum in terms of Brownian self-intersection times
and the growth rate of $V$.
Then, in Theorem \ref{Theorem: Main Examples} we apply this result to
white, fractional, singular, and smooth noises.

\subsection{Comparison with Previous Results and Future Work}

Several techniques have been used thus far to control the variance of linear statistics
for the purpose of proving number rigidity. Prominent examples include
determinantal/Pfaffian or other integrable structure \cite{Buf16,BNQ18,G15,GL18,GP17}, translation
invariance and hyperuniformity \cite{Ghosh17}, and finite-dimensional approximations \cite{RN18}.
By using such methods, number rigidity has been established for
the zeroes of the {\it planar Gaussian Analytic Function}, the {\it Ginibre ensemble}, the {\it Sine-$\beta$ process} (for all $\be>0$),
the {\it Airy-$2$ process}, some {\it Bessel and Gamma point processes}, and more.
While some of the properties used in those papers
are present in some examples of the RSOs in \eqref{eq:Schr},
none of these results provide sufficient conditions that can be applied to
general RSOs.

In closing, we note that the work in this paper raises a number of interesting questions
for future research. Most notably, the number rigidity for the Airy-$\be$ process with $\be\neq2$ is still open.
As mentioned following the statement of Theorem \ref{theorem: informal}, we expect that
proving the rigidity of this point process (along with RSOs whose deterministic
potentials do not satisfy our growth conditions on $V$ in \eqref{Equation: General Growth Bounds}) will require
new insights; see Section \ref{Section: Q of Optimality} for more details.
We leave such questions to future papers.

In another direction, we note that the Feynman-Kac formula can be applied to a
much larger class of random operators than what we consider in this paper,
such as operators acting on $\mbb R^d$ for $d\geq2$
or discrete lattices. Although the technical details of the analysis carried
out in this paper rely crucially on objects that only exist in one dimension
(namely, continuous Brownian local time), we expect that a similar argument
to that outlined in steps {\bf (i)}--{\bf (iv)} of Section \ref{sec:Outline}
can be applied whenever the Feynman-Kac formula holds, and thus
the general methodology developed in this paper has the potential for substantial generalization.
Such extensions will be the subject of future works.

Finally, as mentioned in Section \ref{Sec: Spacial Conditioning}, for many point processes the
understanding of conditional distributions in spacial conditioning is more
sophisticated than number rigidity, such as tolerance in \cite{GP17} or
explicit conditional distributions in \cite{BPreprint,BQ17}. It would be
interesting to see if similar insights in the conditional configurations of
eigenvalues of general RSOs can be obtained. We leave this to future work.

\subsection*{Organization of Paper}
The rest of this paper is organized as follows. In Section~\ref{sec:SetUp},
we introduce the setup of the paper (including the Feynman-Kac formulas
at the heart of our method), state our main results, and discuss their
optimality.
Section~\ref{Section: Self-Intersection} contains estimates on the decay rate (for small time)
of self-intersection local times that are crucial in our method of proof.
In Section~\ref{Section: Variance Estimates}, we combine the estimates in
Section~\ref{Section: Self-Intersection} with our Feynman-Kac formulas to control the variance of exponential
linear statistics, thus proving our main results,
Theorems~\ref{Theorem: Main} and~\ref{Theorem: Main Examples}.
Section~\ref{Section: Optimality} demonstrates that the variance of exponential linear statistics
cannot be used to prove rigidity of the Airy-$2$ process.
Finally, Appendix \ref{Appendix: Some Stoch Analysis} provides an elementary estimate on stochastic analysis.   
 
\subsection*{Acknowledgments}

This work was initiated while the authors were in residence at the
{\it Centre international de rencontres math\'ematiques} (CIRM), in Marseille, France.
The organizers of the conference {\it Integrability and Randomness in Mathematical Physics and Geometry}
(April 2019) and the CIRM staff are gratefully acknowledged for fostering a productive research environment.

The authors thank Ivan Corwin, Vadim Gorin and Mykhaylo Shkolnikov
for insightful discussions and comments, and
Reda Chhaibi for helpful discussions. The authors thank an anonymous
referee for a careful reading of a previous version of this paper, and for many
detailed comments that helped substantially improve the presentation of the paper.

 \section{Setup and Main Results}\label{sec:SetUp}

This section is organized as follows.
In Section~\ref{sec:RigidityDef}, we give reminders for basic notions
regarding number rigidity.
In Section~\ref{sec:Noise}, we state our assumptions regarding the random
perturbation $\xi$ in \eqref{eq:Schr}, and we provide concrete examples
of noises that satisfy these assumptions.
In Section~\ref{sec:Operator}, we discuss the rigorous definition of
the operator $\hat{\mc H}_I$ and its eigenvalue point process.
In Section \ref{sec:FeynmanKac}, we introduce the Feynman-Kac
formulas with which we study exponential linear statistics of $\hat{\mc H}_I$'s
spectrum, including a statement that the linear statistics
in question are finite and well defined.
In Section~\ref{sec:Main}, we state our main results.
Finally, we discuss the optimality of our results and related open problems in Section \ref{Section: Q of Optimality}.

\subsection{Number Rigidity}\label{sec:RigidityDef}

Let $\La$ be a point process on $\mbb R$ (i.e., a random locally finite counting measure
on $\mbb R$). Given a Borel set $A\subset\mbb R$,
we let $\La(A)$ denote the number of points of $\La$ that are inside of $A$, that is,
\[\La(A):=\sum_{\la\in\La}\mbf 1_{\{\la\in A\}}.\]
More generally, for every function $f:\mbb R\to\mbb R$, we use
\[\La(f):=\sum_{\la\in\La}f(\la)\]
to denote the linear statistic associated with $f$.
For any Borel set $A\subset\mbb R$, we let $\ms F_\La(A):=\si\big\{\La(\bar A):\bar A\subset A\big\}$ denote
the $\si$-algebra generated by the configuration of points inside of $A$.

\begin{definition}[\cite{GP17}]
We say that $\La$ is {\bf number rigid} if $\La(A)$ is $\ms F_\La(\mbb R\setminus A)$-measurable
for every bounded Borel set $A\subset\mbb R$.
\end{definition}

We have the following simple sufficient condition for number rigidity:

\begin{proposition}[{\cite{GP17}}]
\label{Proposition: Variance Criterion}
Let $A\subset\mbb R$ be a bounded Borel set. Let $(f_n)_{n\in\mbb N}$ be a sequence of
functions satisfying the following conditions.
\begin{enumerate}
\item Almost surely, $|\La(f_n\mbf 1_{\bar A})|<\infty$ for every $n\in\mbb N$ and $\bar A\subset\mbb R$.
\item $|f_n-1|\to0$ as $n\to\infty$ uniformly on $A$.
\item $\mbf{Var}[\La(f_n)]\to0$ as $n\to\infty$.
\end{enumerate}
Then, $\La(A)$ is $\ms F_{\La}(\mbb R\setminus A)$-measurable.
\end{proposition}

Though Proposition \ref{Proposition: Variance Criterion}
is by now standard in the rigidity literature (e.g.,
\cite[Theorem 6.1]{GP17}), we nevertheless provide
its short proof for the reader's convenience:

\begin{proof}[Proof of Proposition \ref{Proposition: Variance Criterion}]
For every $n$, we can write
\[\La(A)=
\underbrace{\La(f_n)-\mbf E\big[\La(f_n)\big]}_{E_1^{(n)}}+
\underbrace{\La\big((1-f_n)\mbf 1_{A}\big)}_{E_2^{(n)}}
-\underbrace{\big(\La\big(f_n\mbf 1_{\mbb R\setminus A}\big)-\mbf E\big[\La(f_n)\big]\big)}_{E_3^{(n)}}.\]
Since the variance of $\La(f_n)$ vanishes, we can choose a sparse enough subsequence $(n_k)_{k\in\mbb N}$
along which $E_1^{(n_k)}\to0$ almost surely as $k\to\infty$. Next, we note that
\[|E_2^{(n_k)}|\leq\La(A)\left(\sup_{x\in A}|f_{n_k}(x)-1|\right),\]
which vanishes almost surely as $k\to\infty$ because $\La$ is locally finite and $A$ is bounded.
In particular, $E_3^{(n_k)}\to\La(A)$ as $k\to\infty$,
which completes the proof since $E_3^{(n)}$ is $\ms F_{\La}(\mbb R\setminus A)$-measurable for every $n$.
\end{proof}

\subsection{Noise}
\label{sec:Noise}

In this section, we describe the noises $\xi$ considered in this paper.
(Much of the notation in this section and Sections
\ref{sec:Operator} and \ref{sec:FeynmanKac} are
 directly inspired from \cite{GaudreauLamarre}.)
Let $\mr{PC}_c$ denote the set of functions $f:\mbb R\mapsto\mbb R$ that are
c\`adl\`ag and compactly supported.
We begin by introducing the covariance functions that characterize
the noise $\xi$.

\begin{definition}
\label{Definition: Covariance}
Let $\ga$ be an even function on $\mbb R$
or an even Schwartz distribution on $\mr{PC}_c$
(that is, $\langle f,\ga\rangle=\langle \tilde f,\ga\rangle$ for every $f\in\mr{PC}_c$, where $\tilde f(x):=f(-x)$)
such that
\begin{align}
\label{Equation: Semi-Inner-Product}
\langle f,g\rangle_\ga:=\int_{\mbb R^2}f(x)\ga(x-y)g(y)\d x\dd y,\qquad f,g\in \mr{PC}_c
\end{align}
is a semi-inner-product on $\mr{PC}_c$, that is,
\begin{enumerate}
\item \eqref{Equation: Semi-Inner-Product} is finite and well defined for every $f,g\in\mr{PC}_c$;
\item $(f,g)\mapsto\langle f,g\rangle_\ga$ is sesquilinear and symmetric; and
\item $\langle f,f\rangle_\ga\geq0$ for all $f\in\mr{PC}_c$.
\end{enumerate}
We denote the seminorm induced by $\langle \cdot,\cdot\rangle_\ga$ as
\begin{align*}
\|f\|_\ga:=\sqrt{\langle f,f\rangle}_\ga,\qquad f\in\mr{PC}_c.
\end{align*}
We say that $\ga$ is {\bf compactly supported} if there exists a compact
set $A\subset\mbb R$ such that $\langle f,\ga\rangle=0$ whenever $f(x)=0$
for every $x\in A$.
\end{definition}

\begin{remark}
In cases where $\ga$ is not an almost-everywhere-defined function,
the integral over $\ga(x-y)$ in \eqref{Equation: Semi-Inner-Product} may not be well defined.
In such cases, we interpret
\[\langle f,g\rangle_\ga:=\langle f,g*\ga\rangle=\langle f*\tilde g,\ga\rangle=\langle \tilde f*g,\ga\rangle=\langle f*\ga,g\rangle,\]
where $\langle\cdot,\cdot\rangle$ denotes the $L^2$ inner product and $*$ the convolution.
\end{remark}

Throughout this paper, we make the following assumption.

\begin{assumption}
\label{Assumption: Noise}
We assume that there exists
a $\ga$ as in Definition
\ref{Definition: Covariance} such that
\begin{align}
\label{Equation: gamma to Lp Bounds}
\|f\|_{\ga}^2\leq c_\ga\big(\|f\|_{q_1}^2+\cdots+\|f\|_{q_\ell}^2\big),\qquad f\in\mr{PC}_c
\end{align}
for some constant $c_\ga>0$ and $1\leq q_1,\ldots,q_\ell\leq2$, $\ell\in\mbb N$,
where $\|f\|_q:=\left(\int_\mbb R|f(x)|^q\d x\right)^{1/q}$ denotes the usual $L^q$ norm.\\
\end{assumption}

If Assumption \ref{Assumption: Noise} holds, then
it can be shown that there exists a centered Gaussian process
$\Xi:\mbb R\to\mbb R$ such that
\begin{enumerate}
\item almost surely, $\Xi(0)=0$ and $\Xi$ has continuous sample paths;
\item $\Xi$ has stationary increments; and
\item $\Xi$'s covariance is given by
\begin{align}
\label{Equation: Xi Covariance}
\mbf E[\Xi(x)\Xi(y)]=\begin{cases}
\langle\mbf 1_{[0,x)},\mbf 1_{[0,y)}\rangle_\ga&\text{if }x,y\geq0\\
\langle\mbf 1_{[0,x)},-\mbf 1_{[y,0)}\rangle_\ga&\text{if }x\geq0\geq y\\
\langle-\mbf 1_{[x,0)},\mbf 1_{[0,y)}\rangle_\ga&\text{if }y\geq0\geq x\\
\langle\mbf 1_{[x,0)},\mbf 1_{[y,0)}\rangle_\ga&\text{if }0\geq x,y.
\end{cases}
\end{align}
\end{enumerate}
Indeed, the existence of a Gaussian process with covariance
\eqref{Equation: Xi Covariance} follows from standard existence
theorems since $\langle\cdot,\cdot\rangle_\ga$ is a semi-inner-product;
the stationarity of increments follows from the fact that $\langle f,g\rangle_\ga$
remains unchanged if we replace $f$ and $g$ by their translates
$x\mapsto f(x-z)$ and $x\mapsto g(x-z)$ for some $z\in\mbb R$;
and a continuous version can be shown to exists thanks to
Kolmorogov's classical theorem for path continuity.
(We refer to \cite[Remark 2.19 and Section 3.3]{GaudreauLamarre}
for the full details of this argument.)
We think of the noise $\xi$ as the formal derivative of the
continuous stochastic process $\Xi$. More precisely:

\begin{definition}
\label{Definition: Stochastic Integral}
For every $f\in\mr{PC}_c$, we define
\begin{align}
\label{Equation: Stochastic Integral}
\xi(f):=\int_\mbb R f(x)\d\Xi(x),
\end{align}
where $\dd\Xi$ denotes stochastic integration with respect to $\Xi$ interpreted in the
pathwise sense of Karandikar \cite{Karandikar}; we refer to \cite[Section 3.2.1]{GaudreauLamarre}
for the details of this construction.
\end{definition}

\begin{remark}
The properties of $\xi$ as defined above that we need in this paper are that
\begin{enumerate}
\item for every realization of $\Xi$, the map $\xi:\mr{PC}_c\to\mbb R$
is measurable with respect to the uniform topology; and
\item $f\mapsto\xi(f)$ is a centered Gaussian process
on $\mr{PC}_c$ with covariance
\begin{align}
\label{Equation: Covariance Function}
\mbf E[\xi(f)\xi(g)]=\langle f,g\rangle_\ga.
\end{align}
\end{enumerate}
A proof that \eqref{Equation: Stochastic Integral}
satisfies these properties is the subject of \cite[Section 3.2.1]{GaudreauLamarre}.
\end{remark}

We now present several examples of noises covered by Assumption \ref{Assumption: Noise}.
We refer to Lemma \ref{Lemma: Gamma to Lp Bounds} in this paper for a proof
that the examples below satisfy \eqref{Equation: gamma to Lp Bounds}.

\begin{example}[\bf White\textnormal]
\label{Example: White}
Let $\si>0$ be fixed. We say that $\xi$ is a {\bf white noise} with
variance $\si^2$ if $\ga=\si^2\de_0$, where $\de_0$ denotes the delta Dirac distribution.
In this case, the covariance is simply the $L^2$ inner product
\[\mbf E\big[\xi(f)\xi(g)\big]=\si^2\langle f,g\rangle,\]
and $\xi$ can be constructed as the pathwise stochastic integral
\[\xi(f):=\si\int_\mbb R f(x)\d W(x)\]
with respect to a two-sided Brownian motion $W$.
\end{example}

\begin{example}[\bf Fractional\textnormal]
\label{Example: Fractional}
Let $H\in(\tfrac12,1)$ and $\si>0$ be fixed. We say that $\xi$ is a {\bf fractional noise}
with Hurst parameter $H$ and variance $\si^2$ if 
\[\ga(x):=\si^2H(2H-1)|x|^{2H-2},\]
in which case
\[\mbf E\big[\xi(f)\xi(g)\big]=\si^2H(2H-1)\int_{\mbb R^2}\frac{f(x)g(y)}{|x-y|^{2-2H}}\d x\dd y.\]
This noise can be constructed as the pathwise stochastic integral
\[\xi(f):=\si\int_\mbb R f(x)\d W^H(x),\]
where $W^H$ is a two-sided fractional Brownian motion with Hurst parameter $H$.
\end{example}

\begin{example}[\bf $\bs{L^p}$-Singular\textnormal]
\label{Example: Singular}
Let $1\leq p<\infty$. We say that $\xi$ is an {\bf $\bs{L^p}$-singular
noise} if $\ga$ can be decomposed as
\[\ga=\ga_1+\ga_2,\]
where $\ga_1\in L^p(\mbb R)$, and $\ga_2$ is uniformly bounded.
We can view $L^p$-singular noise as a generalization of
fractional noise, as $\ga_1$ may have point singularities, such as $\ga_1(x)\sim |x|^{-\mf e}$ as $x\to0$
for some $\mf e\in(0,1)$, or $\ga_1(x)\sim (-\log |x|)^{\mf e}$ as $x\to0$ for some $\mf e>0$.
\end{example}

\begin{example}[\bf Bounded\textnormal]
\label{Example: Bounded}
We say that $\xi$ is a {\bf bounded noise} if $\ga$ is uniformly bounded.
In many such cases $\xi$ gives rise to a pointwise-defined
Gaussian process on $\mbb R$ with covariance function $\mbf E[\xi(x)\xi(y)]=\ga(x-y)$,
whence we can simply define
\begin{align}
\label{eq:pointwise}
\xi(f):=\int_\mbb R f(x)\xi(x)\d x.
\end{align}
\end{example}

\subsection{Operator and Eigenvalue Point Process}\label{sec:Operator}

We now discuss the definition of the operator $\hat{\mc H}_I$
and its spectrum. We make the following two
assumptions on the domain/boundary conditions of the operator,
and the deterministic potential $V$:

\begin{assumption}
\label{Assumption: Cases}
We consider three types of domains $I\subset\mbb R$ on which $\hat{\mc H}_I$ acts:
the full space $I=\mbb R$ ({\bf Case 1}), the half-line $I=(0,\infty)$ ({\bf Case 2}),
and the bounded interval $I=(0,b)$ for some $b>0$ ({\bf Case 3}).

In {\bf Case 2}, we consider Dirichlet and Robin boundary conditions
at the origin:
\begin{align}\label{Equation: Case 2 Boundary}
\begin{cases}
f(0)=0&\text{(Dirichlet)}\\
f'(0)+\al f(0)=0&\text{(Robin)}
\end{cases}
\end{align}
where $\al\in\mbb R$ is fixed.

In {\bf Case 3}, we consider the Dirichlet, Robin, and mixed boundary conditions at
the endpoints $0$ and $b$:
\begin{align}\label{Equation: Case 3 Boundary}
\begin{cases}
f(0)=f(b)=0&\text{(Dirichlet)}\\
f'(0)+\al f(0)=-f'(b)+\be f(b)=0&\text{(Robin)}\\
f'(0)+\al f(0)=f(b)=0&\text{(Mixed 1)}\\
f(0)=-f'(b)+\be f(b)=0&\text{(Mixed 2)}\\
\end{cases}
\end{align}
where $\al,\be\in\mbb R$ are fixed.
\end{assumption}

\begin{assumption}\label{Assumption: Potential}
$V:I\to\mbb R$ is bounded below and locally integrable on $I$'s closure.
If $I$ is unbounded (i.e., {\bf Cases 1 \& 2}), then we also assume that
\begin{align*}
\lim_{|x|\to\infty}\frac{V(x)}{\log|x|}=\infty.
\end{align*}
\end{assumption}

We may now provide the following definition for the operator $\hat{\mc H}_I$,
which is a direct application of \cite[Proposition 2.9]{GaudreauLamarre},
and allows for a rigorous interpretation of the deterministic operator
$-\tfrac12\De+V$ plus noise $\xi$ through sesquilinear forms
(see also \cite{BloemendalVirag,FukushimaNakao,Minami,RamirezRiderVirag}):

\begin{proposition}
\label{Prop: operator}
Given a fixed choice of domain $I$, boundary conditions,
and potential $V$ all satisfying Assumptions \ref{Assumption: Cases}
and \ref{Assumption: Potential}, let $\mc E$ denote the sesquilinear form
of the corresponding deterministic Schr\"odinger operator $-\tfrac12\De+V$,
and let $D(\mc E)\subset L^2(I)$ be the associated form domain.
(We refer to \cite[Definition 2.6]{GaudreauLamarre} for a precise
statement of these objects in all cases outlined in Assumption
\ref{Assumption: Cases} and \ref{Assumption: Potential}, and to \cite[Section 7.5 and Example 7.5.3]{Simon} for
the standard operator theoretic terminology used here.)

Suppose that Assumption \ref{Assumption: Noise} holds,
and let $\xi$ be as in Definition \ref{Definition: Stochastic Integral}.
With probability one,
there exists a unique self-adjoint operator $\hat{\mc H}_I$
with dense domain $D(\hat{\mc H}_I)\subset L^2$
such that
\begin{enumerate}
\item $D(\hat{\mc H}_I)\subset D(\mc E);$
\item $\langle f,\hat{\mc H}_Ig\rangle=\mc E(f,g)+\xi(fg)$ for every $f,g\in D(\hat{\mc H}_I)$; and
\item $\hat{\mc H}_I$ has compact resolvent.
\end{enumerate}
\end{proposition}

\begin{remark}
Implicit in the statement of Proposition \ref{Prop: operator}
is the claim that the noise $\xi$ can be suitably extended to
products of functions in the form domain $D(\mc E)$. As argued in
\cite[Remark 2.7]{GaudreauLamarre}, this is not a problem.
\end{remark}

With this result in hand, we immediately obtain the following definition
of $\hat{\mc H}_I$'s spectrum by the variational principle
(e.g., \cite[Theorems XIII.2 and XIII.64]{ReedSimon}):

\begin{corollary}
\label{cor: point process}
Under the same hypotheses and notations as Proposition \ref{Prop: operator},
there exists a random orthonormal basis $(\Psi_k)_{k\in\mbb N}$
of $L^2(I)$ and a point process $\La=(\La_k)_{k\in\mbb N}$
on the real line $\mbb R$ such that, almost surely,
\begin{enumerate}
\item $-\infty<\La_1\leq\La_2\leq\La_3\leq\cdots\nearrow+\infty$; and
\item for every $k\in\mbb N$,
\[\La_k=\inf_{\substack{f\in D(\mc E),~\|f\|_2=1\\
f\perp \Psi_1,\ldots\Psi_{k-1}}}\mc E(f,f)+\xi(f^2),\]
where $\Psi_k$ achieves the above infimum.
\end{enumerate}
\end{corollary}

\subsection{Semigroup and Feynman-Kac Formula}\label{sec:FeynmanKac}

We now discuss the semigroup theory of the operator defined in
Proposition \ref{Prop: operator}, and argue that exponential
statistics of its eigenvalue point process defined in Corollary
\ref{cor: point process} can be studied with a Feynman-Kac
formula.
Before we can do this,
we must introduce some stochastic processes that form the basis of the
Feynman-Kac formulas that we use:

\begin{definition}\label{Definition: Local Time Etc}
We use $B$ to denote a standard Brownian motion taking values in $\mbb R$,
$X$ to denote a reflected standard Brownian motion taking values in $(0,\infty)$,
and $Y$ to denote a reflected standard Brownian motion taking values in $(0,b)$.
Throughout this paper, we use $Z$ to denote one of these three processes,
depending on which case in Assumption \ref{Assumption: Cases} is being considered, that is
\begin{align}
\label{Equation: Z}
Z=\begin{cases}
B&\text{({\bf Case 1})}\\
X&\text{({\bf Case 2})}\\
Y&\text{({\bf Case 3})}.
\end{cases}
\end{align}

For every $t>0$ and $x,y\in I$, we denote by
\[Z^x:=\big(Z|Z(0)=x\big)\]
the process started at $x$, and we denote the bridge process from $x$ to $y$ in time $t$ by
\[Z^{x,y}_t:=\big(Z|Z(0)=x\text{ and }Z(t)=y\big).\]
We sometimes use $\mbf E^x$ and $\mbf E^{x,y}_t$ to denote the expected value with respect to the law of
$Z^x$ and $Z^{x,y}_t$, respectively.

We denote the Gaussian kernel by
\begin{align}
\label{Equation: Gaussian Kernel}
\ms G_t(x):=\frac{\mr e^{-x^2/2t}}{\sqrt{2\pi t}},\qquad t>0,~x\in\mbb R.
\end{align}
We denote the transition kernel of $Z$ by $\Pi_Z$, that is, for every $t>0$ and $x,y\in I$
\begin{align*}
\Pi_Z(t;x,y):=\begin{cases}
\ms G_t(x-y)&\text{({\bf Case 1})}\\
\ms G_t(x-y)+\ms G_t(x+y)&\text{({\bf Case 2})}\\
\sum_{z\in2b\mbb Z\pm y}\ms G_t(x-z)&\text{({\bf Case 3})}.
\end{cases}
\end{align*}

For any $0\leq s\leq t$, we let $a\mapsto L^a_{[s,t]}(Z)$ ($a\in I$) denote the continuous version
of the local time of $Z$ (or its conditioned versions) collected on $[s,t]$, i.e.,
\begin{align}
\label{Equation: Local Time Definition}
\int_s^tf\big(Z(u)\big)\d u=\int_I L^a_{[s,t]}(Z)f(a)\d a=\langle L_{[s,t]}(Z),f\rangle
\end{align}
for any measurable function $f:I\to\mbb R$
(see, e.g., \cite[Chapter VI, Corollary 1.6 and Theorem 1.7]{RevuzYor}
for the existence and continuity of local times).
We use the shorthand $L_t(Z):=L_{[0,t]}(Z)$.

As a matter of convention, if $Z=X$ or $Y$, then we distinguish the boundary local time from the above, which we denote as
\[\mf L^c_{[s,t]}(Z):=\lim_{\eps\to0}\frac1{2\eps}\int_s^t\mbf 1_{\{c-\eps<Z(u)<c+\eps\}}\d u\]
for $c\in\partial I$ (i.e., $c=0$ if $Z=X$ or $c\in\{0,b\}$ if $Z=Y$), also with the shorthand $\mf L^c_t(Z):=\mf L^c_{[0,t]}(Z)$.
We refer to \cite[Chapter VI, Corollary 1.9]{RevuzYor} for the relation between this
quantity and the local time as defined in \eqref{Equation: Local Time Definition}.
\end{definition}

We are now finally in a position to state our Feynman-Kac formulas.

\begin{definition}\label{Definition: Feynman-Kac}
In {\bf Cases 2 \& 3}, let us define the quantities $\bar\al$ and $\bar\be$ as
\[\bar\al:=\begin{cases}
-\infty&\text{({\bf Case 2}, Dirichlet)}\\
\al&\text{({\bf Case 2}, Robin)}
\end{cases}
\qquad
(\bar\al,\bar\be):=\begin{cases}
(-\infty,-\infty)&\text{({\bf Case 3}, Dirichlet)}\\
(\al,\be)&\text{({\bf Case 3}, Robin)}\\
(\al,-\infty)&\text{({\bf Case 3}, Mixed 1)}\\
(-\infty,\be)&\text{({\bf Case 3}, Mixed 2)}\\
\end{cases}\]
where $\al,\be\in\mbb R$ are as in \eqref{Equation: Case 2 Boundary} and \eqref{Equation: Case 3 Boundary}.
For every $t>0$ and $x,y\in I$, we define the random kernel
\[{\hat K}(t;x,y):=\begin{cases}
\Pi_B(t;x,y)\mbf E^{x,y}_t\big[\mr e^{-\langle L_t(B),V\rangle-\xi(L_t(B))}\big]&\text{({\bf Case 1})}
\vspace{5pt}\\
\Pi_X(t;x,y)\mbf E^{x,y}_t\big[\mr e^{-\langle L_t(X),V\rangle-\xi(L_t(X))+\bar\al\mf L^{0}_t(X)}\big]&\text{({\bf Case 2})}
\vspace{5pt}\\
\Pi_Y(t;x,y)\mbf E^{x,y}_t\big[\mr e^{-\langle L_t(Y),V\rangle-\xi(L_t(Y))+\bar\al\mf L^{0}_t(Y)+\bar\be \mf L^b_t(Y)}\big]&\text{({\bf Case 3})}
\end{cases}\]
where we assume that the noise $\xi$ is independent of $B$, $X$, or $Y$; hence the expected value $\mbf E^{x,y}_t$
is with respect to $B^{x,y}_t$, $X^{x,y}_t$, or $Y^{x,y}_t$, conditional on $\xi$. We denote by $\hat K(t)$ the random integral operator on $L^2(I)$ with the above kernel.
\end{definition}

\begin{remark}
If $\xi$ can be realized as a pointwise-defined measurable map on $\mbb R$,
then it follows from \eqref{eq:pointwise} and \eqref{Equation: Local Time Definition} that
\[\langle L_t(Z),V\rangle+\xi\big(L_t(Z)\big)=\int_0^tV\big(Z(s)\big)+\xi\big(Z(s)\big)\d s.\]
Thus, in this case $\hat K(t)$ corresponds to the Feynman-Kac representation of the semigroup
generated by the classically well-defined operator $\hat{\mc H}_I:=-\frac12\De+V+\xi$ with the appropriate
boundary condition (e.g., \cite{ChungZhao,Papanicolaou,SimonSemigroups,Sznitman},
or \cite[Theorem 5.4]{GaudreauLamarre} and references therein for a unified statement).
\end{remark}

\begin{remark}
Since we use the continuous version of Brownian local time,
for every $t>0$,
$L_t(Z)$ is an element of $\mr{PC}_c$ almost surely.
Consequently, the term $\xi(L_t(Z))$ in $\hat K(t)$'s definition
is well defined in the sense of Definition \ref{Definition: Stochastic Integral}.
The facts that the functions $(x,y)\mapsto\hat K(t;x,y)$
and $x\mapsto \hat K(t;x,x)$
are measurable on $I\times I$ and $I$ respectively and that $\hat K(t)\in L^2(I\times I)$
are proved in \cite[Theorem 2.23 and Appendix A]{GaudreauLamarre}.
\end{remark}

\begin{remark}
In cases where $\bar\al$ or $\bar\be$ are not finite, we use the conventions
$\mr e^{-\infty}:=0$ and
\[-\infty\cdot \mf L^c_t(Z):=\begin{cases}
0&\text{if }\mf L^c_t(Z)=0\\
-\infty&\text{if }\mf L^c_t(Z)>0.
\end{cases}\]
Thus, for any $c\in\partial I$,
if we let $\tau_c(Z):=\inf\{t\geq0:Z(t)=c\}$ denote the first hitting time of $c$, then we can interpret
$\mr e^{-\infty\cdot \mf L^c_t(Z)}:=\mbf 1_{\{\tau_c(Z)>t\}}$.
\end{remark}

The following result is a direct consequence of \cite[Theorem 2.23]{GaudreauLamarre}
(see also \cite{GaudreauLamarreShkolnikov,GorinShkolnikov}).

\begin{proposition}
\label{prop: F-K}
Suppose that the same hypotheses as Proposition \ref{Prop: operator} hold,
and let $\La=(\La_k)_{k\in\mbb N}$
denote $\hat{\mc H}_I$'s spectrum, as per
Corollary \ref{cor: point process}. For every $t>0$,
\begin{align}
\label{Equation: Trace Formula}
0\leq\mr{Tr}[\mr e^{-t\hat{\mc H}_I}]=\sum_{k=1}^\infty\mr e^{-t\La_k}=\mr{Tr}[\hat K(t)]=\int_I\hat K(t;x,x)\d x<\infty
\qquad\text{almost surely}.
\end{align}
In particular, exponential linear statistics of the form $x\mapsto\mr e^{-t x}$ are well defined in the
point process $\La$ for all $t>0$, and can be computed explicitly using the kernels in Definition \ref{Definition: Feynman-Kac}.
\end{proposition}

\subsection{Main Result}\label{sec:Main}

Our main result is as follows.

\begin{theorem}
\label{Theorem: Main}
Suppose that Assumptions \ref{Assumption: Noise}, \ref{Assumption: Cases}, and \ref{Assumption: Potential} are satisfied,
and let $\hat{\mc H}_I$ be as in Proposition \ref{Prop: operator}.
In {\bf Case 3}, $\hat{\mc H}_{(0,b)}$'s spectrum is always number rigid.
In {\bf Cases 1 \& 2} (i.e., $\hat{\mc H}_\mbb R$ or $\hat{\mc H}_{(0,\infty)}$),
if $\mf d>1$ is such that
\begin{align}
\label{Equation: d Exponent}
\limsup_{t\to0}t^{-\mf d}\left(\sup_{x\in I}\mbf E^x\Big[\|L_t(Z)\|_\ga^{2\theta}\Big]^{1/\theta}\right)<\infty
\end{align}
for every positive $\theta$,
then $\hat{\mc H}_I$'s spectrum is number rigid if the following
growth condition on $V$ holds:
\begin{align}
\label{Equation: General Growth Bounds}
\begin{cases}
\displaystyle
\lim_{|x|\to\infty}\frac{V(x)}{|x|^{2/(2\mf d-1)}}=\infty&\text{(if $\ga$ is compactly supported)}
\vspace{5pt}\\
\displaystyle
\lim_{|x|\to\infty}\frac{V(x)}{|x|^{2/(\mf d-1)}}=\infty&\text{(otherwise)}.
\end{cases}
\end{align}
\end{theorem}

\begin{remark}
If we assume \eqref{Equation: gamma to Lp Bounds},
then \eqref{Equation: d Exponent} always holds
with at least
\begin{align}
\label{Equation: minimal d exponent}
\mf d\geq1+1/\max\{q_1,\ldots,q_\ell\}.
\end{align}
(i.e.,
combine the bound \eqref{Equation: gamma to Lp Bounds}
with \eqref{Equation: Lp Local Time Scaling}; see \eqref{Equation: Gamma to Lp specific} for the details).
In particular, under our assumptions, Theorem \ref{Theorem: Main}
always provides a nontrivial sufficient condition for number rigidity in
{\bf Cases 1 \& 2}.
We nevertheless state the general condition \eqref{Equation: d Exponent}
in Theorem \ref{Theorem: Main} instead of \eqref{Equation: minimal d exponent}, since it is sometimes possible to
find $\mf d>1+1/\max\{q_1,\ldots,q_\ell\}$ such that \eqref{Equation: d Exponent}
holds, and thus prove number rigidity for a larger class of potentials
(see, for example, the case of fractional noise in \eqref{Equation: Gamma to Lp specific fractional}).
\end{remark}

From this theorem, we obtain the following corollary, which specializes
\eqref{Equation: d Exponent} and \eqref{Equation: General Growth Bounds} to the four examples of noises considered
earlier.

\begin{theorem}
\label{Theorem: Main Examples}
Let $\xi$ be one of the four types of noises considered in Examples
\ref{Example: White}--\ref{Example: Bounded}. Then,
\eqref{Equation: d Exponent} holds with
\begin{align}
\label{Equation: d Exponent Examples}
\mf d:=\begin{cases}
3/2&\text{(white noise)}\\
1+H&\text{(fractional noise with index $H\in(\tfrac12,1)$)}\\
2-1/2p&\text{($L^p$-singular noise with $p\geq1$)}\\
2&\text{(bounded noise)}.
\end{cases}
\end{align}
In particular, under Assumptions \ref{Assumption: Cases} and \ref{Assumption: Potential},
 in {\bf Cases 1 \& 2}
$\hat{\mc H}_I$'s spectrum is number rigid if the following sufficient conditions on $V$ are satisfied.
\begin{enumerate}
\item \textnormal{({\bf White})} If $\xi$ is a white noise, then
\begin{equation}\label{eq:Vwhite}
\lim_{|x|\to\infty}\frac{V(x)}{|x|}=\infty.
\end{equation}

\item \textnormal{({\bf Fractional})}
If $\xi$ is a fractional noise with Hurst index $H\in(\tfrac12,1)$, then
\begin{equation}\label{eq:Vfractiona}
\lim_{|x|\to\infty}\frac{V(x)}{|x|^{2/H}}=\infty.
\end{equation}

\item \textnormal{({\bf$\bs{L^p}$-Singular})}
If $\xi$ is an $L^p$-singular noise, then
\begin{equation}\label{eq:VPolSing}
\begin{cases}
\displaystyle
\lim_{|x|\to\infty}\frac{V(x)}{|x|^{2p/(3p-1)}}=\infty&\text{(if $\ga$ is compactly supported)}
\vspace{5pt}\\
\displaystyle
\lim_{|x|\to\infty}\frac{V(x)}{|x|^{4p/(2p-1)}}=\infty&\text{(otherwise)}.
\end{cases}
\end{equation}

\item \textnormal{({\bf Bounded})}
If $\xi$ is a bounded noise, then
\begin{equation}\label{eq:VBound}
\begin{cases}
\displaystyle
\lim_{|x|\to\infty}\frac{V(x)}{|x|^{2/3}}=\infty&\text{(if $\ga$ is compactly supported)}
\vspace{5pt}\\
\displaystyle
\lim_{|x|\to\infty}\frac{V(x)}{|x|^2}=\infty&\text{(otherwise)}.
\end{cases}
\end{equation}
\end{enumerate}
\end{theorem}

Theorem \ref{Theorem: Main} is proved in Section \ref{Section: Variance Estimates}.
The main technical ingredient in this proof is Theorem \ref{Theorem: Variance Bound},
which provides quantitative upper bounds on the variance of the linear statistic
$\sum_k\mr e^{-t\La_k}$ as $t\to0$ using the identity \eqref{Equation: Trace Formula}.
The result then follows from an application of
Proposition \ref{Proposition: Variance Criterion} with test functions of the form
$f_n(x)=\mr e^{-t_nx}$ with $t_n\to0$, by proving that
\[\lim_{n\to\infty}\mbf{Var}\big[\La(f_n)\big]=\lim_{n\to\infty}\mbf{Var}\big[\mr{Tr}[\hat K(t_n)]\big]=0\]
under the conditions stated in Theorem \ref{Theorem: Main}. Theorem \ref{Theorem: Main Examples}
is proved in Section \ref{Section: Proof of Main Corollary}.

\subsection{Questions of Optimality}
\label{Section: Q of Optimality}

\subsubsection{Two Examples}

The growth conditions \eqref{Equation: General Growth Bounds} raise natural questions concerning the optimality of
Theorem \ref{Theorem: Main}. For instance, when $\xi$ is a white noise, it is known that
the super-linear condition $V(x)/|x|\to\infty$ in Theorem \ref{Theorem: Main Examples} is not necessary
for the number rigidity of $\La$.

\begin{proposition}[{\cite{Buf16}}]
\label{Proposition: Optimality 1}
Let $\xi_2$ be a white noise with variance $1/2$.
Let us denote the operator
\begin{align}
\label{Equation: Airy 2 Operator}
\hat{\mc H}^{(2)}_{(0,\infty)}:=-\tfrac12\De+\tfrac x2+\xi_2,
\end{align}
with Dirichlet boundary condition at zero.
$\hat{\mc H}^{(2)}_{(0,\infty)}$'s spectrum is number rigid.
\end{proposition}

Indeed, one may recognize $\hat{\mc H}^{(2)}_{(0,\infty)}$ as the stochastic Airy operator with parameter $\be=2$ (up to a multiple of $1/2$),
whose spectrum forms a determinantal point process (e.g., \cite{RamirezRiderVirag,TracyWidom}) known as the Airy-$2$ process. By using
this integrable structure, Bufetov showed in \cite[Section 3.2]{Buf16} that $\hat{\mc H}^{(2)}_{(0,\infty)}$'s spectrum is
number rigid. In the following proposition (proved in Section~\ref{Section: Optimality}), we demonstrate how exponential linear statistics
fail to show the rigidity of the Airy-$2$ process, and thus \eqref{eq:Vwhite} is the best general sufficient condition for white noise one
can obtain with the method of this paper:

\begin{proposition}
\label{Proposition: Optimality 2}
With $\hat{\mc H}^{(2)}_{(0,\infty)}$ as in \eqref{Equation: Airy 2 Operator}, it holds that
\[\lim_{t\to0}\mbf{Var}\big[\mr{Tr}[\mr e^{-t\hat{\mc H}^{(2)}_{(0,\infty)}}]\big]=(4\pi)^{-1}.\]
\end{proposition}

We also note the following simple example, which shows that our superquadratic
condition in \eqref{eq:VBound} for bounded noise with general $\ga$ is optimal,
and provides an example of a random Schr\"odinger operator whose spectrum
is not number rigid.

\begin{example}
Let $g$ be a standard Gaussian random variable, and suppose that $\xi(x)=g$
for all $x\in\mbb R$. In our terminology, $\xi$ is a bounded noise with non-compactly-supported
covariance function $\ga(x)=1$ for all $x\in\mbb R$. Consider the operator
\begin{align}
\label{Equation: Quantum HO}
\hat{\mc H}^{(\mr{HO})}_\mbb Rf(x):=-\tfrac12f''(x)+x^2f(x)+\xi(x)f(x),
\end{align}
acting on the whole space $\mbb R$.
It is known that the deterministic operator $-\tfrac12\De+x^2$, which is usually called the
quantum harmonic oscillator, has a spectrum of the form $\{c_1k+c_2\}_{k\in\mbb N}$
for some constants $c_1,c_2>0$ (e.g., \cite[Chapter 2, Proposition 2.2 (ii)]{Tak}). In particular,
the spectrum of \eqref{Equation: Quantum HO} consists of the randomly shifted semilattice
$\{c_1k+c_2+g\}_{k\in\mbb N}$, which is clearly not number rigid.
\end{example}

\subsubsection{Some Open Problems}

In light of Propositions \ref{Proposition: Optimality 1} and \ref{Proposition: Optimality 2},
it would be interesting to better understand the conditions
under which the spectrum of one-dimensional continuous RSOs are number rigid,
leading to the following open problem.

\begin{problem}
\label{Problem: Full Characterization}
Suppose that Assumptions \ref{Assumption: Noise} and \ref{Assumption: Potential} hold.
Given a fixed noise $\xi$, characterize the potentials $V$ such that $\hat{\mc H}_I$'s
spectrum is number rigid in {\bf Cases 1 \& 2}.
\end{problem}

A second problem of interest would be to uncover optimal conditions
under which the variance of $\mr{Tr}[\hat K(t)]$ vanishes as $t\to0$.

\begin{problem}
\label{Problem: Semigroup Characterization}
Suppose that Assumptions \ref{Assumption: Noise} and \ref{Assumption: Potential} hold.
Given a fixed noise $\xi$, characterize the potentials $V$ such that
\[\lim_{t\to0}\mbf{Var}\big[\mr{Tr}[\hat K(t)]\big]=0\]
in {\bf Cases 1 \& 2}.
\end{problem}

Owing to Proposition \ref{Proposition: Optimality 2}, the following conjecture concerning
Problem \ref{Problem: Semigroup Characterization} in the case of white noise seems natural.

\begin{conjecture}
Let $\xi$ be a white noise.
Suppose that Assumption \ref{Assumption: Potential} holds.
In {\bf Cases 1 \& 2}, if there exists $\ka,\nu>0$ such that $V(x)\leq\ka|x|+\nu$ for all $x\in I$, then
\[\liminf_{t\to0}\mbf{Var}\big[\mr{Tr}[\hat K(t)]\big]>0.\]
\end{conjecture}

\section{Self-Intersection Local Time}
\label{Section: Self-Intersection}

As mentioned in the introduction (and as evidenced by \eqref{Equation: d Exponent}),
controlling the small-$t$ decay rate of self-intersection local times is a crucial ingredient in the proof of our results.
To this effect, in this section, our purpose is to provide one of the main technical ingredient
that we use to establish \eqref{Equation: d Exponent}: Namely, for every $1\leq q\leq 2$, there exists a nonnegative random variable $R_q$ with finite exponential moments in a neighbourhood of zero such that
\begin{align}
\label{Equation: Lp Local Time Scaling}
\sup_{x\in I}\|L_t(Z^x)\|_q^2\leq t^{1+1/q}R_q\qquad\text{for all $t\in (0,1)$,}
\end{align}
where the inequality in \eqref{Equation: Lp Local Time Scaling} is understood
in the sense of stochastic domination.
(Recall that for any two random variables $X$ and $Y$, $X$ is said to be \emph{stochastically dominated} by $Y$ if $\mbf{E}[f(X)]\leq \mbf{E}[f(Y)]$ for any nondecreasing function $f$.
This is equivalent to saying that there exists a random variable $Z$ with the same distribution as $Y$
 such that $X\leq Z$ almost surely; see, e.g., \cite[Theorem 1]{KKO}). We refer to the proof of Theorem \ref{Theorem: Main Examples} in Section \ref{Section: Proof of Main Corollary}
for an explanation of how \eqref{Equation: Lp Local Time Scaling} is used to prove
\eqref{Equation: d Exponent Examples}.
\begin{proposition}
\label{Proposition: Lp Local Time Scaling}
Define $\mathcal{L}^{\sup}:= \sup_{a\in\mbb R}L^a_1(B^0)$.
Let us denote the maximum and minimum of the Brownian motion $B^x$ as
\begin{align}
\label{Equation: Brownian min and max}
M^x(t):=\sup_{s\in[0,t]}B^x(s)
\qquad\text{and}\qquad
m^x(t):=\inf_{s\in[0,t]}B^x(s).
\end{align}
For $q=1$, define $R_q:=1$, and for $q\in (1,2]$, let
\begin{align}
\label{Equation: Rq Variable}
R_q:= \begin{cases}
2^{2(q-1)/q}\|L_1(B^{0})\|^2_{q} & (\textbf{Cases 1 \& 2})\\
c\big(\mathcal{L}^{\sup}\big)^{2(1-1/q)}
+c\bigg(2(\mathcal{L}^{\sup})^2+2\big(M^0(1)-m^0(1)\big)^2\bigg)^{2(1-1/q)} & (\textbf{Case 3}),
\end{cases}
\end{align}
where $c>0$ in {\bf Case 3} is a deterministic constant that only depends on the size of the interval $I=(0,b)$ and $q$.
Then, \eqref{Equation: Lp Local Time Scaling} holds for all $q\in [1,2]$ with $R_q$ shown above.  
\end{proposition}

\begin{proof}
Recall that, thanks to \eqref{Equation: Local Time Definition}, $\|L_t(Z)\|_1=t$.
Thus, if $q=1$, then \eqref{Equation: Lp Local Time Scaling} holds trivially with
$R_q=1$.

We therefore only need to prove \eqref{Equation: Lp Local Time Scaling}
for $q\in(1,2]$. We argue case by case. Let us begin with \textbf{Case 1} which corresponds to $I =\mathbb{R}$. If we couple $B^x=x+B^0$ for all $x\in \mathbb{R}$, then straightforward
changes of variables with a Brownian scaling imply that
\begin{align}
\|L_t(B^x)\|_q^2=\|L_t(B^0)\|_q^2
\deq t\left(\int_{\mbb R} L^{t^{-1/2}a}_1(B^0)^q\d a\right)^{2/q}
=t^{1+1/q}\|L_1(B^0)\|_q^2 \label{Equation: Local Time Inequality}
\end{align}
for every $q>1$.
According to \cite[Theorem 4.2.1]{Chen},
for every $q>1$ there exists some $c>0$ such that
\begin{align}
\label{Equation: Self-Intersection Tails}
\mbf P\big[\|L_1(B^0)\|_q^2>u\big]=\mr e^{-cu^{q/(q-1)}(1+o(1))},\qquad u\to\infty.
\end{align}
This shows $\|L_1(B^0)\|_q^2$ has exponential moments for $1<q\leq 2$.
Thus, in {\bf Case 1} we have \eqref{Equation: Lp Local Time Scaling} with $R_q=2^{2(q-1)/q} \|L_1(B^0)\|_q^2$
since $2^{2(q-1)/q}>1$ whenever $q>1$.

Consider now \textbf{Case 2} where $I$ is taken to be $(0,\infty)$ and $X$ is a reflected Brownian motion taking values in $(0,\infty)$. By coupling $X^x(t)=|B^x(t)|$ for all $t>0$, we note that
for every $a>0$, one has $L_t^a(X^x)=L_t^a(|B^x|)=L_t^a(B^x)+L_t^{-a}(B^x).$
Therefore,
\begin{multline*}
\|L_t(X^x)\|_q^2
=\left(\int_0^\infty L^a_t(X^x)^q\d a\right)^{2/q}\\
\leq 2^{2(q-1)/q}\left(\int_0^\infty L^a_t(B^x)^q+L^{-a}_t(B^x)^q\d a\right)^{2/q}
=2^{2(q-1)/q}\|L_t(B^x)\|_q^2.
\end{multline*}
By \eqref{Equation: Local Time Inequality}, the right-hand side of above display is equal in distribution
to
\[t^{1+1/q}2^{2(q-1)/q}\|L_1(B^0)\|_q^2.\] Owing to \eqref{Equation: Self-Intersection Tails}, $R_q=2^{2(q-1)/q}\|L_1(B^0)\|_q^2$ has finite exponential moments
for $1<q\leq 2$,   
thus the proof of \eqref{Equation: Lp Local Time Scaling} in {\bf Case 2} follows.

Finally, consider \textbf{Case 3} where $I$ is an interval $(0,b)$ for some $b>0$ and $Y$ is a reflected Brownian motion
taking values in $(0,b)$.
We note that we can couple the processes $Y^x$ and $B^x$
in such a way that $Y^x$ is obtained by reflecting the path of $B^x$ on the boundary of $(0,b)$,
namely,
\begin{align}
\label{Equation: Coupling of B and Y}
Y^x(t)=\begin{cases}
B^x(t)-2kb&\text{if }B^x(t)\in[2kb,(2k+1)b],\quad k\in\mbb Z,\\
|B^x(t)-2kb|&\text{if }B^x(t)\in[(2k-1)b,2kb],\quad k\in\mbb Z.
\end{cases}
\end{align}
Under this coupling, we observe that for any $z\in(0,b)$, one has
\begin{align}\label{Equation: Interval Local Time}
L_t^z(Y^x)=\sum_{a\in2b\mbb Z\pm z}L^a_t(B^x).
\end{align}
The argument that follows
is inspired from the proof of \cite[Lemma 2.1]{ChenLi}
(see also \cite[Lemma 5.10]{GaudreauLamarre}):
Under \eqref{Equation: Interval Local Time},
\begin{multline*}
\Big(\int_0^b L_t^z(Y^x)^q\d z\Big)^{1/q}
=\Big(\int_0^b \Big(\sum_{k\in2b\mbb Z} L_t^{k+z}(B^x)+L_t^{k-z}(B^x)\Big)^q\d z\Big)^{1/q}\\
\leq2^{(q-1)/q}\sum_{k\in2b\mbb Z}\Big(\int_{-b}^bL^{k+z}_t(B^x)^q\d z\Big)^{1/q}.
\end{multline*}

Recall that $M^x(t)$ and $m^{x}(t)$ are the maximum and minimum of $B^{x}$ in the interval $[0,t]$. In order for $\int_{-b}^bL^{k+z}_t(B^x)^2\d z$ to be nonzero, it must be the case that
$M^x(t)\geq k-b$ and $m^x(t)\leq k+b$, or, equivalently,
$M^x(t)+b\geq k \geq m^x(t)-b$. Thus, for any $q>1$,
\begin{align}
&\sum_{k\in2b\mbb Z}\Big(\int_{-b}^bL^{k+z}_t(B^x)^q\d z\Big)^{1/q}\nonumber\\
&=\sum_{k\in2b\mbb Z}\Big(\int_{-b}^bL^{k+z}_t(B^x)^q\d z\Big)^{1/q}\mbf 1_{\{M^x(t)+b\geq k \geq m^x(t)-b\}}\nonumber\\
&\leq\Big(\sum_{k\in2b\mbb Z}\int_{-b}^bL^{k+z}_t(B^x)^q\d z\Big)^{1/q}\Big(\sum_{k\in2b\mbb Z}\mbf 1_{\{M^x(t)+b\geq k \geq m^x(t)-b\}}\Big)^{1-1/q}\nonumber\\
&=\Big(\int_\mbb RL^a_t(B^x)^q\d a\Big)^{1/q}\Big(\sum_{k\in2b\mbb Z}\mbf 1_{\{M^x(t)+b\geq k \geq m^x(t)-b\}}\Big)^{1-1/q}\nonumber\\
&\leq c_1t^{1/q}\Big(\sup_{a\in\mbb R}L^a_t(B^x)\Big)^{1-1/q}\big(M^x(t)-m^x(t)+c_2\big)^{1-1/q}\nonumber\\
&\leq c_1t^{1/q}\Bigg(c_2^{1-1/q}\big(\sup_{a\in\mbb R}L^a_t(B^x)\big)^{1-1/q}+\Big(\sup_{a\in\mbb R}L^a_t(B^x)\cdot\big(M^x(t)-m^x(t)\big)\Big)^{1-1/q}\Bigg)\label{Equation: EqString}
\end{align}
where $c_1,c_2>0$ only depend on $b$ and $q$: Indeed, the inequality in the third line follows by H\"older's inequality,
the equality in  the fourth line is obtained by noting that $\sum_{k \in 2b\mathbb{Z}}\int^{b}_{-b} L^{a}_t(B^x)^q \d a$ is equal to $\int_\mbb RL^a_t(B^x)^q\d a$,
we get the inequality in the fifth line by noting that $\int_\mbb RL^a_t(B^x)^q\d a$ is bounded by $(\sup_{a\in\mbb R}L^a_t(B^x))^{q-1} \|L_t(B^x)\|_1$ where $\|L_t(B^x)\|_1=t$, and the inequality in the last line follows by bounding $\big(M^x(t)-m^x(t)+c_2\big)^{1-1/q}$ by $\big(M^x(t)-m^x(t)\big)^{1-1/q}+c^{1-1/q}_2$.      

Given that the distributions of the supremum of local time of $B^x$ and the range $M^x(t)-m^x(t)$ are independent of
the starting point $x$, by Brownian scaling, we have that
\begin{align}\label{Equation: Distribution Equivalence 1}
t^{1/q}\big(\sup_{a\in\mbb R}L^a_t(B^x)\big)^{1-1/q}\deq t^{1/2+1/2q}\big(\mc L^{\sup}\big)^{1-1/q}
\end{align}
and
\begin{align}\label{Equation: Distributional Equivalence 2}
t^{1/q}\bigg(\sup_{a\in\mbb R}L^a_t(B^x)\cdot\big(M^x(t)-m^x(t)\big)\bigg)^{\frac{q-1}{q}}
\deq t\,\bigg(\mc L^{\sup}\cdot\big(M^0(1)-m^0(1)\big)\bigg)^{\frac{q-1}{q}}.
\end{align}
Combining \eqref{Equation: EqString} with \eqref{Equation: Distribution Equivalence 1} and \eqref{Equation: Distributional Equivalence 2} shows that $\|L_t(Y^{x})\|^2_q$ is stochastically dominated by the random variable
\[t^{1+1/q}\left(c\big(\mathcal{L}^{\sup}\big)^{2(1-1/q)}
+t^{1-1/q}c\bigg(2(\mathcal{L}^{\sup})^2+2\big(M^0(1)-m^0(1)\big)^2\bigg)^{2(1-1/q)}\right)\]
 where the constant $c>0$ depends only on $b$ and $q$. The right-hand side of the above display is bounded by $t^{1+1/q}R_q$
 in \eqref{Equation: Rq Variable} for {\bf Case 3} for all $t\in (0,1)$. 
Note that there exists $\theta_0>0$ small enough so that
\begin{align}
\label{Equation: Sup Range Exponential Moments}
\mbf E\left[\exp\left(\theta_0 \sup_{a\in\mbb R}L^a_1(B^0)^2\right)\right],\mbf E\left[\mr e^{\theta_0(M^0(1)-m^0(1))^2}\right]<\infty
\end{align}
(e.g., the proof of \cite[Lemma 2.1]{ChenLi} and references therein).
Given that $4(1-1/q)\leq 2$, for $q\in(1,2]$, $R_q$ in \textbf{Case 3} has finite exponential moments
in a neighborhood of zero. This completes the proof of \eqref{Equation: Lp Local Time Scaling} in \textbf{Case 3},
and thus the proof of Proposition \ref{Proposition: Lp Local Time Scaling}.
\end{proof}

\section{Asymptotic Variance Estimates}
\label{Section: Variance Estimates}

In this section, we provide the main technical contributions of this paper,
and use the latter to prove our two main theorems.
The chief result in this direction consists of the following variance upper bounds
for the trace of $\hat K(t)$ as $t\to0$.

\begin{theorem}
\label{Theorem: Variance Bound}
Suppose that Assumptions \ref{Assumption: Noise},
\ref{Assumption: Cases}, and
\ref{Assumption: Potential} hold. Let $\mf d>1$ be as in \eqref{Equation: d Exponent}.
In {\bf Cases 1 \& 2}, assume that there exists $\ka,\nu,\mf a>0$ such that
\begin{align}\label{Equation: Polynomial Potential Condition}
V(x)\geq|\ka x|^{\mf a}-\nu\qquad\text{for every }x\in I.
\end{align}
In {\bf Cases 1 \& 2},
there exists a finite constant $C_\mf a>0$ that only depends on $\mf a$ such that
\begin{align}
\begin{cases}
\label{Equation: Variance Bound Cases 1 and 2}
\displaystyle
\limsup_{t\to0}\frac{\mbf{Var}\big[\mr{Tr}[\hat K(t)]\big]}{t^{\mf d-1/2-1/\mf a}}\leq \frac{C_\mf a}{\ka}&\text{(if $\ga$ is compactly supported)}
\vspace{5pt}\\
\displaystyle
\limsup_{t\to0}\frac{\mbf{Var}\big[\mr{Tr}[\hat K(t)]\big]}{t^{\mf d-1-2/\mf a}}\leq\frac{C_\mf a}{\ka^2}&\text{(otherwise)}.
\end{cases}
\end{align}
In {\bf Case 3}, one has
\begin{align}
\label{Equation: Variance Bound Case 3}
\limsup_{t\to0}\frac{\mbf{Var}\big[\mr{Tr}[\hat K(t)]\big]}{t^{\mf d-1}}<\infty.
\end{align}
\end{theorem}

The remainder of this section is organized as follows:
In Sections \ref{Section: Rigidity Corollary} and \ref{Section: Proof of Main Corollary},
we use Theorem \ref{Theorem: Variance Bound} to prove
our main results, namely, Theorems \ref{Theorem: Main} and
\ref{Theorem: Main Examples} respectively.
Next, in Section \ref{Section: Variance Bound Strategy}, we prove Theorem \ref{Theorem: Variance Bound}.
In order to not interrupt the flow of the argument,
most of the more technical results used to prove Theorems
\ref{Theorem: Main}, \ref{Theorem: Main Examples}, and
\ref{Theorem: Variance Bound} are stated without proof  in Sections
\ref{Section: Rigidity Corollary}--\ref{Section: Variance Bound Strategy};
the technical results in question are then proved Sections 
\ref{Section: Seminorm Bounds} to \ref{Section: Variance Bounds Final}.

\subsection{Proof of Theorem \ref{Theorem: Main}}
\label{Section: Rigidity Corollary}

Let $(t_n)_{n\in\mbb N}$ be a sequence of positive numbers such that $t_n\to0$ as
$n\to0$. For every $n\in\mbb N$, let us define the test function $f_n(x):=\mr e^{-t_nx}$.
This sequence of functions converges to 1 uniformly on compact sets.
Moreover,
by \eqref{Equation: Trace Formula},
\[\La(f_n)=\sum_{k=1}^\infty\mr e^{-t_n\La_k}=\mr{Tr}[\hat K(t_n)]<\infty.\]
Hence, by Proposition \ref{Proposition: Variance Criterion},
to prove that $\La$ is number rigid, it suffices to show that
\begin{align}
\label{Equation: Variance to Zero}
\lim_{n\to\infty}\mbf{Var}\big[\mr{Tr}[\hat K(t_n)]\big]=0.
\end{align}
We now prove that \eqref{Equation: Variance to Zero} holds
under the conditions stated in Theorem~\ref{Theorem: Main}.

In \textbf{Case 3}, \eqref{Equation: Variance to Zero}
is an immediate consequence of \eqref{Equation: Variance Bound Case 3}
since $\mf d>1$ implies that $O(t^{\mf d-1})=o(1)$ as $t\to0$.
Consider then \textbf{Cases 1 \& 2}. If we know that
$V(x)/|x|^\mf a\to\infty$, then for every $\ka>0$, we can choose
$\nu_\ka>0$ large enough so that $V(x)\geq|\ka x|^\mf a-\nu_\ka$ for every $x\in I$.
As per \eqref{Equation: General Growth Bounds}, we choose
\[\begin{cases}
\mf a=2/(2\mf d-1)\iff\mf d-1/2-1/\mf a=0&\text{(if $\ga$ is compactly supported)}\\
\mf a=2/(\mf d-1)\iff\mf d-1-2/\mf a=0&\text{(otherwise)},
\end{cases}\]
and thus \eqref{Equation: Variance Bound Cases 1 and 2} yields
\[\limsup_{n\to\infty}\mbf{Var}\big[\mr{Tr}[\hat K(t_n)]\big]
\leq
\begin{cases}
\displaystyle
C_\mf a/\ka&\text{(if $\ga$ is compactly supported)}\\
C_\mf a/\ka^2&\text{(otherwise)}.
\end{cases}\]
Since $\ka>0$ was arbitrary,
we then obtain \eqref{Equation: Variance to Zero} in \textbf{Cases 1 \& 2}
by taking $\ka\to\infty$, thus concluding the proof of Theorem \ref{Theorem: Main}.

\subsection{Proof of Theorem \ref{Theorem: Main Examples}}
\label{Section: Proof of Main Corollary}

We want to prove that \eqref{Equation: d Exponent}
holds with the choices of $\mf d>1$ in \eqref{Equation: d Exponent Examples}.
Our main tool in proving this is the following lemma.

\begin{lemma}
\label{Lemma: Gamma to Lp Bounds}
There exists a constant $c>0$ (which only depends on $\ga$)
such that for every $f\in\mr{PC}_c$ and $t>0$, one has
\begin{align}
\label{Equation: Gamma to Lp Bounds}
\|f\|_\ga^2\leq
\begin{cases}
c\|f\|_2^2&\text{(white noise)}\\
ct^H\big(t^{-1/2}\|f\|_2^2+t^{-1}\|f\|_1^2\big)&\text{(fractional noise with $H\in(\tfrac12,1)$)}\\
c\big(\|f\|_{1/(1-1/2p)}^2+\|f\|_1^2\big)&\text{($L^p$-singular noise with $p\geq1$)}\\
c\|f\|_1^2&\text{(bounded noise)}.
\end{cases}
\end{align}
\end{lemma}

Lemma \ref{Lemma: Gamma to Lp Bounds} is proved in Section
\ref{Section: Seminorm Bounds}, and is a relatively straightforward
consequence of applying Young's convolution inequality to the
semi-inner-product $\langle f,g\rangle_\ga$.
With \eqref{Equation: Gamma to Lp Bounds} in hand, the result
follows directly from a combination of
\eqref{Equation: Lp Local Time Scaling}
and dominated convergence:
On the one hand,
if it holds that $\|f\|_\ga^2\leq c_\ga\big(\|f\|_{q_1}^2+\cdots+\|f\|_{q_\ell}^2\big)$
for some $1\leq q_i\leq 2$ and $\ell\in\mbb N$, then
an application of \eqref{Equation: Lp Local Time Scaling}
yields
\begin{align}
\label{Equation: Gamma to Lp specific}
\sup_{x\in I}\mbf E^x\big[\|L_t(Z)\|_\ga^{2\theta}\big]^{1/\theta}
=O\left(\sum_{i=1}^\ell t^{1+1/q_i}\mbf E[R_{q_i}^\theta]^{1/\theta}\right)=O(t^{1+1/\max\{q_1,\ldots,q_\ell\}})
\end{align}
as $t\to0$. In the case of white, $L^p$-singular, and bounded noise,
this immediately yields \eqref{Equation: d Exponent Examples} thanks to \eqref{Equation: Gamma to Lp Bounds}.
On the other hand, in the case of fractional noise,
an application of \eqref{Equation: Lp Local Time Scaling} and \eqref{Equation: Gamma to Lp Bounds} yields
the following asymptotic as $t\to0$, concluding the proof of \eqref{Equation: d Exponent Examples}:
\begin{align}
\label{Equation: Gamma to Lp specific fractional}
\sup_{x\in I}\mbf E^x\big[\|L_t(Z)\|_\ga^{2\theta}\big]^{1/\theta}
= O\left(t^H\big(t^{-1/2+3/2}\mbf E[R_2^\theta]^{1/\theta}+t^{-1+2}\mbf E[R_1^\theta]^{1/\theta}\big)\right)
=O(t^{1+H}).
\end{align}

\subsection{Proof of Theorem \ref{Theorem: Variance Bound}}
\label{Section: Variance Bound Strategy}

We divide the proof of Theorem~\ref{Theorem: Variance Bound} into three steps. In the first step (Section~\ref{Section: Variance Formula 2}), we derive an integral formula of $\mbf {Var}[\mathrm{Tr}[\hat{K}(t)]]$. The second step (Section~\ref{Section: Technical Lemmas}) provides upper bounds on the different components of the integral formula. Those upper bounds are summarized in few technical lemmas whose proofs are relegated to Section~\ref{Section: Variance Formula}-\ref{Section: Variance Bounds Final}.
The third and final step (Section~\ref{Section: Conclusion of Proof}) completes the proof of Theorem~\ref{Theorem: Variance Bound} by combining the ingredients of Section~\ref{Section: Technical Lemmas} with the integral formula of Section~\ref{Section: Variance Formula 2}.

\subsubsection{Step 1. Variance Formula}\label{Section: Variance Formula 2}

We begin by introducing some notational shortcuts used
throughout this section to improve readability:

\begin{notation}
For the remainder of Section \ref{Section: Variance Estimates}, we use $C,c>0$ to denote constants independent of
$\ka$, $\nu$, $\mf a$ and $t$ whose precise values may change from one equation to the next,
and we use $C_\mf a>0$ to denote such constants that depend only on $\mf a$.
\end{notation}

\begin{notation}
\label{Notation: Exponential Trace Variance Shorthands}
Let $Z$ be as in \eqref{Equation: Z}, and let $\bar Z$ be an independent copy of $Z$.
For every $t>0$, we define the following random functions:
For $(x,y)\in I^2$,
\begin{align*}
\mc A_t(x,y)&:=-\langle L_t(Z^{x,x}_t)+L_t(\bar Z^{y,y}_t),V\rangle,\\
\mc B_t(x,y)&:=
\begin{cases}
0&\text{({\bf Case 1})}\\
\bar\al\mf L^{0}_t(X^{x,x}_t)+\bar\al\mf L^{0}_t(\bar X^{y,y}_t)&\text{({\bf Case 2})}\\
\bar\al\mf L^{0}_t(Y^{x,x}_t)+\bar\be \mf L^b_t(Y^{x,x}_t)+\bar\al\mf L^{0}_t(\bar Y^{y,y}_t)+\bar\be \mf L^b_t(\bar Y^{y,y}_t)&\text{({\bf Case 3})},
\end{cases}\\
\mc C_t(x,y)&:=\frac{\|L_t(Z^{x,x}_t)\|_\ga^2+\|L_t(\bar Z^{y,y}_t)\|_\ga^2}{2},\\
\mc D_t(x,y)&:=\langle L_t(Z^{x,x}_t),L_t(\bar Z^{y,y}_t)\rangle_\ga,\\
\mc P_t(x,y)&:=\Pi_Z(t;x,x)\Pi_Z(t;y,y).
\end{align*}
\end{notation}

Our variance formula is as follows:

\begin{lemma}
\label{Lemma: Variance Formula}
Following Notation \ref{Notation: Exponential Trace Variance Shorthands},
it holds that
\begin{align}
\label{Equaton: Variance Formula}
\mbf{Var}\big[\mr{Tr}[\hat K(t)]\big]=\int_{I^2}\mc P_t(x,y)\,\mbf E\left[\mr e^{ (\mc A_t+\mc B_t+\mc C_t)(x,y)}
\left(\mr e^{\mc D_t(x,y)}-1\right)\right]\d x\dd y.
\end{align}
\end{lemma}

Lemma \ref{Lemma: Variance Formula} is proved in Section~\ref{Section: Variance Formula} using the Feynman-Kac
formula in Proposition \ref{prop: F-K}.

\subsubsection{Step 2. Technical Results}
\label{Section: Technical Lemmas}

By a combination of applying H\"older's inequality to \eqref{Equaton: Variance Formula}
and bounding $\mc P_t(x,y)$ uniformly in $x,y\in I$ using the right-hand side of \eqref{Equation: Transition Bounds},
we obtain the following upper bound for $t\in(0,1]$:
\begin{multline}
\label{Equaton: Variance Formula 2}
\mbf{Var}\big[\mr{Tr}[\hat K(t)]\big]\leq
Ct^{-1}\int_{I^2}
\mbf E\left[\mr e^{4\mc A_t(x,y)}\right]^{1/4}
\mbf E\left[\mr e^{4\mc B_t(x,y)}\right]^{1/4}\\
\times\mbf E\left[\mr e^{4\mc C_t(x,y)}\right]^{1/4}
\mbf E\left[\left(\mr e^{\mc D_t(x,y)}-1\right)^4\right]^{1/4}\d x\dd y.
\end{multline}
At this point, the proof of Theorem \ref{Theorem: Variance Bound} is reduced
to controlling the $t\to0$ asymptotics of the four terms involving
$\mc A_t$, $\mc B_t$, $\mc C_t$, and $\mc D_t$
on the right-hand side of \eqref{Equaton: Variance Formula 2}.
We now state the technical results we use for this purpose.
Our first such result states that the contributions of
$\mc B_t$ and $\mc C_t$ to
\eqref{Equaton: Variance Formula 2} are uniformly bounded for small $t$:

\begin{lemma}
\label{Lemma: B & C Bounds}
For any $\theta>0$,
\begin{align}
\label{Equation: B Bound}
\limsup_{t\to0}\sup_{(x,y)\in I^2}\mbf E\left[\mr e^{\theta \mc B_t(x,y)}\right]&\leq C,\\
\limsup_{t\to0}\sup_{(x,y)\in I^2}\mbf E\left[\mr e^{\theta\mc C_t(x,y)}\right]&\leq C.
\label{Equation: C Bound}
\end{align}
\end{lemma}

Lemma \ref{Lemma: B & C Bounds} is proved in Section \ref{Section: Uniformly Bounded Terms}.
One of the main technical ingredients in the proof of this result is the estimate \eqref{Equation: Lp Local Time Scaling},
together with a midpoint sampling trick that allows to extend the latter (which concerns the unconditioned process $Z^x$)
to the bridge processes $Z^{x,x}_t$ (see \eqref{Equation: Tower+Doob}--\eqref{Equation: Midpoint Trick 3}
for the details).

Our second and third technical results 
concern the decay rate of the expectation involving $\mc D_t$. On the one hand,
the following result explains the distinction between general $\ga$
and compactly supported $\ga$ in
Theorem \ref{Theorem: Variance Bound} for {\bf Cases 1 \& 2}:

\begin{lemma}
\label{Lemma: D Bound 1}
Let $\theta>0$ be arbitrary.
Let $K>0$ be such that $\ga$ is supported on the compact interval $[-K,K]$
(that is, $\langle f,\ga\rangle=0$ for every $f$ that vanishes in $[-K,K]$).
In {\bf Case 1},
\[\Big(\mbf E\Big[\big|\mr e^{\mc D_t(x,y)}-1\big|^\theta\Big]\Big)^{1/\theta}
\leq C\mr e^{-\frac{(|x-y|-K)^2}{2ct}}\Big(\mbf E\Big[\big|\mr e^{\mc D_t(x,y)}-1\big|^{2\theta}\Big]\Big)^{1/2\theta}\]
for all $x,y\in\mbb R$. In {\bf Case 2}, for every $x,y>0$, one has
\[\Big(\mbf E\Big[\big|\mr e^{\mc D_t(x,y)}-1\big|^\theta\Big]\Big)^{1/\theta}
\leq C\left(\mr e^{-\frac{(|x-y|-K)^2}{2ct}}+\mr e^{-\frac{(|x+y|-K)^2}{2ct}}\right)\Big(\mbf E\Big[\big|\mr e^{\mc D_t(x,y)}-1\big|^{2\theta}\Big]\Big)^{1/2\theta}.\]
\end{lemma}

Lemma \ref{Lemma: D Bound 1} is proved in Section \ref{Section: Compactly Supported},
and its proof consists of a formalization of the following simple heuristic: The farther apart $x$ and $y$
are from each other, the more likely it is that the supports of $L_t(Z^{x,x}_t)$ and $L_t(\bar Z^{y,y}_t)$
are separated by a distance of at least $K>0$, in which case the semi-inner-product
$\mc D_t(x,y)=\langle L_t(Z^{x,x}_t),L_t(\bar Z^{y,y}_t)\rangle_\ga$ vanishes if $\ga$ is supported
in $[-K,K]$. On the other hand, the following result provides an estimate on the decay rate of $\mr e^{\mc D_t(x,y)}-1$
as $t\to0$, and explains the appearance of the assumption \eqref{Equation: d Exponent}
in the statement of Theorem \ref{Theorem: Main}:

\begin{lemma}
\label{Lemma: D Bound 2}
Let $\mf d>1$ be as in \eqref{Equation: d Exponent}.
For any $\theta>0$,
\[\limsup_{t\to0}t^{-\mf d}\,\sup_{(x,y)\in I^2}\Big(\mbf E\Big[\big|\mr e^{\mc D_t(x,y)}-1\big|^\theta\Big]\Big)^{1/\theta}\leq C.\]
\end{lemma}

Lemma \ref{Lemma: D Bound 2} is proved in Section \ref{Section: Vanishing Term}.
Our final technical result
concerns the $t\to0$ asymptotics of the term involving $\mc A_t$ in \eqref{Equaton: Variance Formula 2}:

\begin{lemma}
\label{Lemma: Final Estimates}
Let $\mf a,\ka>0$ be as in \eqref{Equation: Polynomial Potential Condition}.
One the one hand, it holds that
\begin{align}
\label{Equation: Final Estimate General}
\limsup_{t\to0}t^{2/\mf a}\int_{I^2}\mbf E\left[\mr e^{4\mc A_t(x,y)}\right]^{1/4}\d x\dd y\leq\frac{C_{\mf a}}{\ka^2}
\end{align}
in {\bf Cases 1 \& 2}.
On the other hand, for every $c,K>0$, one has
\begin{align}
\label{Equation: Final Estimate Compact 1}
\limsup_{t\to0}t^{-1/2+1/\mf a}\int_{\mbb R^2}\mbf E\left[\mr e^{4\mc A_t(x,y)}\right]^{1/4}\,\mr e^{-\frac{(|x-y|-K)^2}{2ct}}\d x\dd y\leq\frac{C_{\mf a}}{\ka}
\end{align}
in {\bf Case 1}; and in {\bf Case 2}, it holds that
\begin{align}
\label{Equation: Final Estimate Compact 2}
\limsup_{t\to0}t^{-1/2+1/\mf a}\int_{(0,\infty)^2}\mbf E\left[\mr e^{4\mc A_t(x,y)}\right]^{1/4}\left(\mr e^{-\frac{(|x-y|-K)^2}{2ct}}+\mr e^{-\frac{(|x+y|-K)^2}{2ct}}\right)\d x\dd y\leq\frac{C_{\mf a}}{\ka}.
\end{align}
\end{lemma}

Lemma \ref{Lemma: Final Estimates} is proved in
Section \ref{Section: Variance Bounds Final}, and relies on a formalization
of the heuristic that, if we assume \eqref{Equation: Polynomial Potential Condition},
then we expect that $\mbf E\left[\mr e^{4\mc A_t(x,y)}\right]^{1/4}=O(\mr e^{2\nu t}\mr e^{-t(|\ka x|^\mf a+|\ka y|^\mf a)})$
as $t\to0$.

\subsubsection{Step 3. Conclusion of Proof}\label{Section: Conclusion of Proof}

We now use the technical lemmas stated in Section \ref{Section: Technical Lemmas}
to conclude the proof of Theorem \ref{Theorem: Variance Bound}.
Thanks to \eqref{Equaton: Variance Formula 2} and Lemma \ref{Lemma: B & C Bounds},
we have that
\begin{align}
\label{Variance Formula 3}
\mbf{Var}\big[\mr{Tr}[\hat K(t)]\big]=O\left(t^{-1}\int_{I^2}
\mbf E\left[\mr e^{4\mc A_t(x,y)}\right]^{1/4}
\mbf E\left[\left(\mr e^{\mc D_t(x,y)}-1\right)^4\right]^{1/4}\d x\dd y\right)
\end{align}
as $t\to0$, where the constant in $O$ is independent of
all parameters. We now control the right-hand-side of \eqref{Variance Formula 3}
on a case-by-case basis.

Let us begin with {\bf Case 1}. In the case of general $\ga$
(i.e., not necessarily compactly supported), it follows from
Lemma \ref{Lemma: D Bound 2} and \eqref{Equation: Final Estimate General}
that
\begin{multline}
\label{Equation: Wrap-up Case 1 General Gamma}
\limsup_{t\to0}t^{-\mf d+2/\mf a}\int_{\mbb R^2}
\mbf E\left[\mr e^{4\mc A_t(x,y)}\right]^{1/4}
\mbf E\left[\left(\mr e^{\mc D_t(x,y)}-1\right)^4\right]^{1/4}\d x\dd y\\
\leq\limsup_{t\to0}\left(t^{2/\mf a}\int_{\mbb R^2}
\mbf E\left[\mr e^{4\mc A_t(x,y)}\right]^{1/4}\d x\dd y\right)
\left(t^{-\mf d}\,\sup_{(x,y)\in \mbb R^2}\Big(\mbf E\Big[\big|\mr e^{\mc D_t(x,y)}-1\big|^4\Big]\Big)^{1/4}\right)\leq\frac{C_\mf a}{\ka^2}.
\end{multline}
When combined with \eqref{Variance Formula 3},
this yields \eqref{Equation: Variance Bound Cases 1 and 2} in {\bf Case 1} for general $\ga$.
If $\ga$ is compactly supported in some interval $[-K,K]$,
then it follows from Lemma \ref{Lemma: D Bound 1} that
\begin{multline}
\label{Equation: Bound for Compact gamma}
\int_{\mbb R^2}\mbf E\left[\mr e^{4\mc A_t(x,y)}\right]^{1/4}
\mbf E\left[\left(\mr e^{\mc D_t(x,y)}-1\right)^4\right]^{1/4}\d x\dd y\\
\leq\left(\int_{\mbb R^2}\mbf E\left[\mr e^{4\mc A_t(x,y)}\right]^{1/4}\,\mr e^{-\frac{(|x-y|-K)^2}{2ct}}\d x\dd y\right)\left(\sup_{(x,y)\in \mbb R^2}\Big(\mbf E\Big[\big|\mr e^{\mc D_t(x,y)}-1\big|^8\Big]\Big)^{1/8}\right).
\end{multline}
At this point, by arguing as in \eqref{Equation: Wrap-up Case 1 General Gamma}
(except that we replace the estimate \eqref{Equation: Final Estimate General}
with \eqref{Equation: Final Estimate Compact 1}),
we obtain that
\[\limsup_{t\to0}t^{-\mf d-1/2+1/\mf a}\int_{\mbb R^2}
\mbf E\left[\mr e^{4\mc A_t(x,y)}\right]^{1/4}
\mbf E\left[\left(\mr e^{\mc D_t(x,y)}-1\right)^4\right]^{1/4}\d x\dd y\leq\frac{C_\mf a}{\ka}.\]
Combining this with \eqref{Variance Formula 3}
yields \eqref{Equation: Variance Bound Cases 1 and 2} in {\bf Case 1} for compactly supported $\ga$,
concluding the proof of Theorem \ref{Theorem: Variance Bound} in {\bf Case 1}.

The proof of Theorem \ref{Theorem: Variance Bound} in {\bf Case 2} follows from the
same steps used in {\bf Case 1}, except that we replace \eqref{Equation: Bound for Compact gamma}
with the corresponding bound given by Lemma \ref{Lemma: D Bound 1} in {\bf Case 2},
and that we replace an application of \eqref{Equation: Final Estimate Compact 1}
with  \eqref{Equation: Final Estimate Compact 2}.

We now conclude the proof of Theorem \ref{Theorem: Variance Bound} with \textbf{Case 3}.
By Assumption \ref{Assumption: Potential}, $V$ is bounded below,
i.e., there exists some $c\geq0$ such that $V(x)\geq-c$ for every $x$. Thus,
\begin{align}
\label{Equation: Case 3 Lower Bound Simple}
\sup_{x,y\in(0,b)}\mbf E\left[\mr e^{4\mc A_t(x,y)}\right]^{1/4}\leq\mr e^{2ct}\leq C
\end{align}
for $t\in(0,1]$.
Since the integral in \eqref{Variance Formula 3} is over the bounded
domain $I^2=(0,b)^2$ in {\bf Case 3},
\eqref{Equation: Variance Bound Case 3} then follows from a direct application
of Lemma \ref{Lemma: D Bound 2} and \eqref{Equation: Case 3 Lower Bound Simple} to \eqref{Variance Formula 3},
concluding the proof of Theorem \ref{Theorem: Variance Bound}.

\subsection{Seminorm Bounds:
Proof of Lemma \ref{Lemma: Gamma to Lp Bounds}}
\label{Section: Seminorm Bounds}

We provide a case-by-case argument.
If $\xi$ is a white noise, then up to a constant $\|\cdot\|_\ga=\|\cdot\|_2$, so the
result is immediate.

For fractional noise, up to a constant, we have that
\[\|f\|_\ga^2\leq\int_{\mbb R^2} |f(a)\ga(a-b)f(b)|\d a\dd b=\int_{\mbb R^2} \frac{|f(a)f(b)|}{|a-b|^{2-2H}}\d a\dd b.\]
By applying the change of variables $(a,b)\mapsto t^{1/2}(a,b)$ to the right-hand side of this equation, we obtain
\[t\int_{\mbb R^2} \frac{|f(t^{\frac{1}{2}}a)f(t^{\frac{1}{2}}b)|}{|t^{\frac{1}{2}}(a-b)|^{2-2H}}\d a\dd b
=t^H\int_{\mbb R^2} \frac{|f(t^{\frac{1}{2}}a)f(t^{\frac{1}{2}}b)|}{|a-b|^{2-2H}}\d a\dd b.\]
Next, we write
\begin{align}
\label{Equation: Fractional Split Bound}
\int_{\mbb R^2} \frac{|f(t^{\frac{1}{2}}a)f(t^{\frac{1}{2}}b)|}{|a-b|^{2-2H}}\d a\dd b
=\Big(\int_{\{|b-a|<1\}}+\int_{\{|b-a|\geq1\}}\Big)\frac{|f(t^{\frac{1}{2}}a)f(t^{\frac{1}{2}}b)|}{|a-b|^{2-2H}}\d a\dd b.
\end{align}
On the one hand, by Young's convolution inequality (e.g., \cite{Young}), the first integral (integral over $\{|b-a|<1\}$) in the r.h.s. of \eqref{Equation: Fractional Split Bound}
is bounded above by
\[\left(\int_{-1}^1\frac{1}{|z|^{2-2H}}\d z\right)\left(\int_{\mbb R}f(t^{\frac{1}{2}}a)^2\d a\right)=\left(\int_{-1}^1\frac{1}{|z|^{2-2H}}\d z\right)t^{-\frac{1}{2}}\|f\|_2^2,\]
where the right-hand side comes from the change of variables $a\mapsto t^{-\frac{1}{2}}a$.
On the other hand, by the same change of variables, the second integral (integral over $\{|b-a|\geq1\}$) is bounded by
\[\left(\int_{\mbb R}|f(t^{\frac{1}{2}}a)|\d a\right)^2=t^{-1}\|f\|_1^2.\]
 Substituting these two bounds in the r.h.s. of \eqref{Equation: Fractional Split Bound} yields the desired bound on $\|f\|^2_{\gamma}$ in the case of fractional noise with Hurst parameter $H \in (\frac{1}{2},1)$. 

Let $\xi$ be an $L^p$-singular noise with decomposition $\ga=\ga_1+\ga_2$.
Then, the bound on $\|f\|_\ga^2$ follows from the following use Young's inequality,
\begin{align*}
\int_{\mbb R^2} |f(a)\ga(a-b)f(b)|\d a\dd b
&=\int_{\mbb R^2} |f(a)\ga_1(a-b)f(b)|\d a\dd b
+\int_{\mbb R^2} |f(a)\ga_2(a-b)f(b)|\d a\dd b\\
&\leq\|\ga_1\|_p\|f\|_q^2+\|\ga_2\|_\infty\|f\|_1^2
\end{align*}
where $\frac1q+\frac1q+\frac1p=2$,
or equivalently,
$q=1/(1-\frac1{2p})$. 

Finally, if $\ga$ is bounded, then
\[\int_{\mbb R^2} |f(a)\ga(a-b)f(b)|\d a\dd b\leq\|\ga\|_\infty\|f\|_1^2,\]
concluding the proof of Lemma \ref{Lemma: Gamma to Lp Bounds}, and thus also of Theorem \ref{Theorem: Main Examples}.

\subsection{Variance Formula:
Proof of Lemma \ref{Lemma: Variance Formula}}
\label{Section: Variance Formula}

We only prove Lemma \ref{Lemma: Variance Formula} in {\bf Case 1}, since the other cases follow from exactly the same argument.
By \eqref{Equation: Covariance Function}, we know that $\mbf E[\mr e^{-\xi(f)}]=\mr e^{\frac12\|f\|^2_\ga}$
for all $f\in\mr{PC}_c$.
Thus, it follows from Fubini's theorem and
\eqref{Equation: Trace Formula} that
\begin{multline*}
\mbf E\big[\mr{Tr}[\hat K(t)]\big]=\int_\mbb R\Pi_B(t;x,x)\mbf E^{x,x}_t\left[\mr e^{-\langle L_t(B),V\rangle}\mbf E_\xi\left[\mr e^{-\xi(L_t(B))}\right]\right]\d x\\
=\int_\mbb R\Pi_B(t;x,x)\mbf E^{x,x}_t\left[\mr e^{-\langle L_t(B),V\rangle+\frac12\|L_t(B)\|_\ga^2}\right]\d x,
\end{multline*}
where $\mbf E_\xi$ denotes the expectation with respect to $\xi$,
conditional on $B$.
Via another application of Fubini, we get
\begin{align}
\Big(\mbf E\big[\mr{Tr}[\hat K(t)]\big]\Big)^2
&=\int_{\mbb R^2}\mc P_t(x,y)\,\mbf E\Big[\mr e^{-\langle L_t(B^{x,x}_t)+L_t(\bar B^{y,y}_t),V\rangle + \frac12\|L_t(B^{x,x}_t)\|_\ga^2
+\frac12\|L_t(\bar B^{y,y}_t)\|_\ga^2}\Big]\d x\dd y \nonumber\\
\label{Equation: Int One Step}&=\int_{\mbb R^2}\mc P_t(x,y)\,\mbf E\Big[\mr e^{(\mc A_t+\mc B_t+ \mc C_t)(x,y)}\Big]\d x\dd y  
\end{align}
where $\bar B^{y,y}_t$ is a Brownian bridge independent of $B^{x,x}_t$. A similar computation yields
\begin{align*}
\mbf E\Big[\big(\mr{Tr}[\hat K(t)]\big)^2\Big]=\int_{\mbb R^2}\mc P_t(x,y)\,\mbf E\Big[\mr e^{-\langle L_t(B^{x,x}_t)+L_t(\bar B^{y,y}_t),V\rangle}
\cdot\mbf E_\xi\Big[\mr e^{-\xi(L_t(B^{x,x}_t)+L_t(\bar B^{y,y}_t))}\Big]\Big]\d x\dd y.
\end{align*}
Given that $\xi(L_t(B^{x,x}_t)+L_t(\bar B^{y,y}_t))$ is Gaussian with mean zero and variance
\[\|L_t(B^{x,x}_t)\|_\ga^2+\|L_t(\bar B^{y,y}_t)\|_\ga^2+2\langle L_t(B^{x,x}_t),L_t(\bar B^{y,y}_t)\rangle_\ga,\]
we may now write 
\[\mbf E\Big[\big(\mr{Tr}[\hat K(t)]\big)^2\Big]= \int_{\mbb R^2}\mc P_t(x,y)\,\mbf E\Big[\mr e^{(\mc A_t+\mc B_t+ \mc C_t + \mc D_t)(x,y)}\Big]\d x\dd y. \]
Finally, the result follows by subtracting \eqref{Equation: Int One Step} from $\mbf E\Big[\big(\mr{Tr}[\hat K(t)]\big)^2\Big]$ in the above display.

\subsection{Uniformly Bounded Terms:
Proof of Lemma \ref{Lemma: B & C Bounds}}
\label{Section: Uniformly Bounded Terms}

We begin with \eqref{Equation: C Bound}.
By Independence,
\[\mbf E\left[\mr e^{\theta\mc C_t(x,y)}\right]=\mbf E^{x,x}_t\left[\mr e^{\frac\theta2\|L_t(Z)\|_\ga^2}\right]\mbf E^{y,y}_t\left[\mr e^{\frac\theta2\|L_t(Z)\|_\ga^2}\right]\]
As it turns out, \eqref{Equation: C Bound}
follows from \eqref{Equation: Lp Local Time Scaling}. The trick that we
use to prove this makes several other appearances in this paper:
Since the exponential function is nonnegative, for every $\theta>0$,
it follows from the tower property and the Doob $h$-transform that
\begin{multline}
\label{Equation: Tower+Doob}
\mbf E^{x,x}_{t}\left[\mr e^{\theta\|L_{t}(Z)\|_\ga^2}\right]
=\mbf E\left[\mbf E^{x,x}_{t}\Big[\mr e^{\theta\|L_{t}(Z)\|_\ga^2}\big|Z^{x,x}_{t}(t/2)\Big]\right]\\
=\int_I\mbf E^{x,x}_{t}\Big[\mr e^{\theta\|L_{t}(Z)\|_\ga^2}\big|Z^{x,x}_{t}(t/2)=y\Big]\frac{\Pi_Z(t/2;x,y)\Pi_Z(t/2;y,x)}{\Pi_Z(t;x,x)}\d y.
\end{multline}

If we condition on $Z^{x,x}_{t}(t/2)=y$, then the paths $\big(Z^{x,x}_{t}(s):0\leq s\leq t/2\big)$ and
$\big(Z^{x,x}_{t}(s):t\leq s\leq t/2\big)$ are independent and have respective distributions $Z^{x,y}_{t/2}$
and $Z^{y,x}_{t/2}$. Since $\Pi_Z$ is a symmetric kernel, the time-reversed process
$s\mapsto Z^{y,x}_{t/2}(t-s)$ ($0\leq s\leq t$)
is equal in distribution to $Z^{x,y}_{t/2}$.
Since local time is additive,
$\mbf E^{x,x}_{t}\Big[\mr e^{\theta\|L_{t}(Z)\|_\ga^2}\big|Z^{x,x}_{t}(t/2)=y\Big]$ is equal to $\mbf E^{x,x}_{t}\Big[\mr e^{\theta\|L_{t/2}(Z)+L_{[t/2,t]}(Z)\|_\ga^2}\big|Z^{x,x}_{t}(t/2)=y\Big]$. Moreover,
\begin{align}
\nonumber
&\mbf E^{x,x}_{t}\Big[\mr e^{\theta\|L_{t/2}(Z)+L_{[t/2,t]}(Z)\|_\ga^2}\big|Z^{x,x}_{t}(t/2)=y\Big]\\
\nonumber
&\leq\mbf E^{x,x}_{t}\Big[\mr e^{2\theta(\|L_{t/2}(Z)\|_\ga^2+\|L_{[t/2,t]}(Z)\|_\ga^2)}\big|Z^{x,x}_{t}(t/2)=y\Big]\\
\label{Equation: Midpoint Trick 2}
&=\mbf E^{x,y}_{t/2}\Big[\mr e^{2\theta\|L_{t/2}(Z)\|_\ga^2}\Big]^2\leq\mbf E^{x,y}_{t/2}\Big[\mr e^{4\theta\|L_{t/2}(Z)\|_\ga^2}\Big],
\end{align}
where the inequality in the second line follows from a combination of the triangle inequality
(since $\|\cdot\|_\ga$ is a seminorm) and $(z+\bar z)^2\leq2(z^2+\bar z^2)$,
the first equality in \eqref{Equation: Midpoint Trick 2} follows from the fact that local time is invariant
with respect to time reversal,
and the second inequality in \eqref{Equation: Midpoint Trick 2} follows from Jensen's inequality.

At this point, if we let
\begin{align}\label{Equation: Midpoint Transition Bound}
\mf s(Z):=\sup_{t\in(0,1]}\sup_{x,y\in I}\frac{\Pi_Z(t/2;y,x)}{\Pi_Z(t;x,x)},
\end{align}
which we know is finite thanks to \eqref{Equation: Transition Bounds},
then, owing to the last inequality of \eqref{Equation: Midpoint Trick 2}, we obtain
\begin{align}
\label{Equation: Midpoint Trick 3}
\mbf E^{x,x}_{t}\left[\mr e^{\theta\|L_{t}(Z)\|_\ga^2}\right]
\leq \mf s(Z)\int_I\mbf E^{x,y}_{t/2}\left[\mr e^{4\theta\|L_{t/2}(Z)\|_\ga^2}\right]\Pi_Z(t/2;x,y)\d y=\mf s(Z)\,\mbf E^x\big[\mr e^{4\theta\|L_{t/2}(Z)\|_\ga^2}\big]
\end{align}
for every $t\leq 1$.
In conclusion, to prove \eqref{Equation: C Bound},
it is enough to show that
\begin{align*}
\limsup_{t\to0}\sup_{x\in I}\mbf E^x\left[\mr e^{\theta\|L_t(Z)\|_\ga^2}\right]\leq C.
\end{align*}
This follows directly from a combination of \eqref{Equation: gamma to Lp Bounds},
\eqref{Equation: Lp Local Time Scaling}, and dominated convergence.

We now prove \eqref{Equation: B Bound}.
In \textbf{Case 1} the result is trivial.
In \textbf{Case 2}, by using essentially the same argument leading up to \eqref{Equation: Midpoint Trick 3},
we have that
\[\mbf E^{x,x}_t\left[\mr e^{\theta\mf L^0_t(X)}\right]\leq C\,\mbf E^x\left[\mr e^{2\theta\mf L^0_{t/2}(X)}\right].\]
By coupling $X^x(s)=|B^x(s)|$ for all $s\geq0$,
this yields
\[\mbf E^{x,x}_t\left[\mr e^{\theta\mf L^0_t(X)}\right]\leq C\,\mbf E^x\left[\mr e^{2\theta\mf L^0_{t/2}(B)}\right],\]
where we define
\[\mf L^a_t(B)
:=\lim_{\eps\to0}\frac1{2\eps}\int_0^t\mbf 1_{\{a-\eps<B(s)<a+\eps\}}\d s\\
=\lim_{\eps\to0}\frac1{2\eps}\int_0^t\mbf 1_{\{a-\eps<|B(s)|<a+\eps\}}\d s.\]
for any $a\in\mbb R$.
Thus, by a straightforward application of H\"older's inequality, it suffices to prove that
\[\limsup_{t\to0}\sup_{x\in\mbb R}\mbf E^x\left[\mr e^{\theta\mf L^0_t(B)}\right]\leq C.\]
By Brownian scaling, $\mf L_t^0(B^x)\deq t^{1/2}\mf L_1^0(B^{t^{-1/2}x})$.
By repeating the proof of \cite[Lemma 5.6]{GaudreauLamarre} in its entirety,
we have that
\[\sup_{x\in\mbb R}\mbf E^x\left[\mr e^{\theta t^{1/2}\mf L_1^0(B)}\right]=\mbf E^0\left[\mr e^{\theta t^{1/2}\mf L_1^0(B)}\right]\leq C,\]
and thus the result follows from dominated convergence.

Consider now \textbf{Case 3}.
Once again arguing as in \eqref{Equation: Midpoint Trick 3},
it suffices to prove that
\begin{align}
\label{Equation: Boundary Local Time Case 3}
\limsup_{t\to0}\sup_{x\in(0,b)}\mbf E^x\left[\mr e^{\theta\mf L^c_t(Y)}\right]\leq C,\qquad c\in\{0,b\}.
\end{align}
Recall the coupling of $Y$ and $B$ in \eqref{Equation: Coupling of B and Y}.
Under this coupling, we observe that
\begin{align}\label{Equation: Interval Boundary Local Time}
\mf L_t^c(Y^x)=\begin{cases}
\displaystyle\sum_{a\in2b\mbb Z}\mf L^a_t(B^x)&(c=0)\\
\displaystyle\sum_{a\in b(2\mbb Z+1)}\mf L^a_t(B^x)&(c=b).
\end{cases}
\end{align}
Consider the case $c=0$.
According to \eqref{Equation: Interval Boundary Local Time},
we see that
\[\mf L^0_t(Y^x)\leq\sup_{a\in\mbb R}\mf L^a_t(B^x)\cdot\mf n_t,\]
where $\mf n_t$ counts the number of intervals of the form $[kb,(k+1)b]$ ($k\in\mbb Z$)
such that
\[\inf_{kb\leq a\leq(k+1)b}\mf L^a_t(B^x)>0.\]
It is easy to see that there exists constants $c_1,c_2>0$ that only depend on $b$ such that
for every $t>0$, one has
$\mf n_t\leq c_1\left(M^x(t)-m^x(t)+c_2\right),$
where we denote $M^x$ and $m^x$ as in \eqref{Equation: Brownian min and max}.
By Brownian scaling,
\begin{align*}
\sup_{a\in\mbb R}\mf L^a_t(B^x)\deq t^{1/2}\sup_{a\in\mbb R}\mf L^a_1(B^0),
\end{align*}
and
\[\left(\sup_{a\in\mbb R}\mf L^a_t(B^x)\right)\big(M^x(t)-m^x(t)\big)\\\deq t\left(\sup_{a\in\mbb R}\mf L^a_1(B^0)\right)\big(M^0(1)-m^0(1)\big).\]
By combining the fact that these terms are independent of $x$ with \eqref{Equation: Sup Range Exponential Moments},
we obtain \eqref{Equation: Boundary Local Time Case 3} for $c=0$. The proof for $c=b$ is nearly identical,
thus concluding the proof of \eqref{Equation: B Bound}, and therefore the proof of Lemma \ref{Lemma: B & C Bounds}.

\subsection{Compactly Supported $\ga$:
Proof of Lemma \ref{Lemma: D Bound 1}}
\label{Section: Compactly Supported}

We begin with the claimed bound in
{\bf Case 1}. Since $\ga$ is supported in $[-K,K]$, in order for the quantity
$\mc D_t(x,y)=\langle L_t(B^{x,x}_t),L_t(\bar B^{y,y}_t)\rangle_\ga$
to be nonzero, it must be the case that
\[\begin{cases}
\displaystyle
\max_{0\leq s\leq t}B^{x,x}_t(s)+K\geq\min_{0\leq s\leq t}\bar B^{y,y}_t(s)&\text{(if $x\leq y$)}
\vspace{5pt}\\
\displaystyle
\max_{0\leq s\leq t}\bar B^{y,y}_t(s)+K\geq\min_{0\leq s\leq t}B^{x,x}_t(s)&\text{(if $x\geq y$)}.
\end{cases}\]
Looking at the case where $x\leq y$, this means that
\begin{align}
\nonumber
&\mbf E\bigg[\left|\mr e^{\mc D_t(x,y)}-1\right|^\theta\bigg]^{1/\theta}\\
\nonumber
&=\mbf E\bigg[\mbf 1_{\{\max_{0\leq s\leq t}B^{x,x}_t(s)+K\geq\min_{0\leq s\leq t}\bar B^{y,y}_t(s)\}}\left|\mr e^{\mc D_t(x,y)}-1\right|^\theta\bigg]^{1/\theta}\\
\label{Equation: Compactly Supported Trick}
&\leq\mbf P\left[\max_{0\leq s\leq t}B^{x,x}_t(s)+K\geq\min_{0\leq s\leq t}\bar B^{y,y}_t(s)\right]^{1/{2\theta}}\mbf E\bigg[\left|\mr e^{\mc D_t(x,y)}-1\right|^{2\theta}\bigg]^{1/{2\theta}}
\end{align}
If we apply a Brownian scaling and use the fact that the maxima of brownian bridges have sub-Gaussian tails,
then
\begin{multline*}
\mbf P\left[\max_{0\leq s\leq t}B^{x,x}_t(s)+K\geq\min_{0\leq s\leq t}\bar B^{y,y}_t(s)\right]^{1/{2\theta}}\\
=\mbf P\left[\max_{0\leq s\leq 1}B^{0,0}_1(s)+\max_{0\leq s\leq 1}\bar B^{0,0}_1(s)\geq(y-x-K)/t^{1/2}\right]^{1/{2\theta}}\leq C\mr e^{-\frac{(y-x-K)^2}{2ct}}.
\end{multline*}
A similar bound is obtained when $x\geq y$, which, when combined with \eqref{Equation: Compactly Supported Trick},
concludes the proof of Lemma \ref{Lemma: D Bound 1} in {\bf Case 1}.

We now provide the proof of Lemma \ref{Lemma: D Bound 1} in {\bf Case 2}.
By H\"older's inequality,
\[\mbf E\bigg[\left|\mr e^{\mc D_t(x,y)}-1\right|^\theta\bigg]^{1/\theta}
\leq\mbf P\big[\langle L_t(X^{x,x}_t),L_t(\bar X^{y,y}_t)\rangle_\ga\neq0\big]^{1/{2\theta}}\mbf E\bigg[\left|\mr e^{\mc D_t(x,y)}-1\right|^{2\theta}\bigg]^{1/{2\theta}}.\]
Note that we can couple $X$ and $B$ so that $X^x(t)=|B^x(t)|$
for all $t\geq0$. Then, conditioning on the endpoint corresponds to
\[X^{x,x}_t=\big(|B^x|\,\big|B^x(t)\in\{x,-x\}\big).\]
Using this coupling, it follows from \cite[(5.9)]{GaudreauLamarre}
that for any nonnegative path functional $F$,
\begin{align}
\label{Equation: X and B Coupling inequality}
\mbf E\big[F(X^{x,x}_t)\big]\leq2\mbf E\big[F(|B^{x,x}_t|)\big].
\end{align}
Consequently, we get the further upper bound
\[\mbf P\big[\langle L_t(X^{x,x}_t),L_t(\bar X^{y,y}_t)\rangle_\ga\neq0\big]^{1/{2\theta}}
\leq2^{1/{2\theta}}\mbf P\big[\langle L_t(|B^{x,x}_t|),L_t(|\bar B^{y,y}_t|)\rangle_\ga\neq0\big]^{1/{2\theta}}.\]
Given that $L_t^a(|B^{x,x}_t|)=L_t^a(B^{x,x}_t)+L_t^{-a}(B^{x,x}_t)$ for all $a>0$
and similarly for $\bar B^{y,y}_t$,
we can expand $\langle L_t(|B^{x,x}_t|),L_t(|\bar B^{y,y}_t|)\rangle_\ga$ as the sum
\begin{multline*}
\int_{(0,\infty)^2}L_t^a(B^{x,x}_t)\ga(a-b)L_t^b(\bar B^{y,y}_t)\d a\dd b
+\int_{(0,\infty)^2}L_t^{-a}(B^{x,x}_t)\ga(a-b)L_t^{-b}(\bar B^{y,y}_t)\d a\dd b\\
+\int_{(0,\infty)^2}L_t^{-a}(B^{x,x}_t)\ga(a-b)L_t^b(\bar B^{y,y}_t)\d a\dd b
+\int_{(0,\infty)^2}L_t^a(B^{x,x}_t)\ga(a-b)L_t^{-b}(\bar B^{y,y}_t)\d a\dd b.
\end{multline*}
Let us define the set $\mc S:=(-\infty,0)^2\cup(0,\infty)^2$.
Since $\ga$ is assumed to be even, by a simple change of variables, the first two terms in the above sum add up to
\begin{align}
\label{Equation: Sesquilinear Decomposition 1}
\int_{\mc S}L_t^a(B^{x,x}_t)\ga(a-b)L_t^b(\bar B^{y,y}_t)\d a\dd b,
\end{align}
and the last two terms add up to
\begin{align}
\label{Equation: Sesquilinear Decomposition 2}
\int_{\mc S}L_t^a(B^{x,x}_t)\ga(a-b)L_t^{-b}(\bar B^{y,y}_t)\d a\dd b.
\end{align}

Suppose that $0<x\leq y$.
In order for \eqref{Equation: Sesquilinear Decomposition 1} to be nonzero,
it must be the case that
\[\max_{0\leq s\leq t}B^{x,x}_t(s)+K\geq\min_{0\leq s\leq t}\bar B^{y,y}_t(s),\]
and for \eqref{Equation: Sesquilinear Decomposition 2} to be nonzero, it must be the case that
\[-\min_{0\leq s\leq t}\bar B^{y,y}_t(s)+K\geq\min_{0\leq s\leq t}B^{x,x}_t(s).\]
Thus, by a union bound, followed by Brownian scaling and the fact that Brownian bridge maxima
have sub-Gaussian tails, we see that
\begin{align*}
&\mbf P\big[|\langle L_t(|B^{x,x}_t|),L_t(|\bar B^{y,y}_t|)\rangle_\ga|>0\big]^{1/{2\theta}}\\
&\leq\mbf P\left[\max_{0\leq s\leq 1}B^{0,0}_1(s)+\max_{0\leq s\leq 1}\bar B^{0,0}_1(s)\geq\frac{y-x-K}{t^{1/2}}\right]^{1/{2\theta}}\\
&\hspace{2in}+\mbf P\left[\max_{0\leq s\leq 1}B^{0,0}_1(s)+\max_{0\leq s\leq 1}\bar B^{0,0}_1(s)\geq\frac{x+y-K}{t^{1/2}}\right]^{1/{2\theta}}\\
&\leq C\left(\mr e^{-\frac{(|x-y|-K)^2}{2ct}}+\mr e^{-\frac{(|x+y|-K)^2}{2ct}}\right).
\end{align*}
The same bound holds for $y\leq x$, concluding the proof of Lemma \ref{Lemma: D Bound 1} in {\bf Case 2}.

\subsection{Vanishing Term:
Proof of Lemma \ref{Lemma: D Bound 2}}
\label{Section: Vanishing Term}

By combining the inequality $|\mr e^z-1|\leq\mr e^{|z|}-1\leq|z|\mr e^{|z|}$ $(z\in\mbb R)$
with $|\mc D_t(x,y)|\leq\tfrac12\big(\|L_t(Z^{x,x}_t)\|_\ga^2+\|L_t(\bar Z^{y,y}_t)\|_\ga^2\big),$
and applying the triangle inequality, we see that
\begin{align*}
 \Big(\mbf E\Big[\big|\mr e^{\mc D_t(x,y)}-1\big|^\theta\Big]\Big)^{1/\theta} &\leq C\Big(\mbf E\left[\|L_t(Z^{x,x}_t)\|_\ga^{2\theta}\mr e^{(\theta/2)(\|L_t(Z^{x,x}_t)\|_\ga^2+\|L_t(\bar Z^{y,y}_t)\|_\ga^2)}\right]^{1/\theta}\\
 &+\mbf E\left[\|L_t(\bar Z^{y,y}_t)\|_\ga^{2\theta}\mr e^{(\theta/2)(\|L_t(Z^{x,x}_t)\|_\ga^2+\|L_t(\bar Z^{y,y}_t)\|_\ga^2)}\right]^{1/\theta}\Big).
\end{align*}
By using independence of $Z$ and $\bar Z$ and applying
H\"older's inequality, the right-hand side of the above inequality is bounded by
\begin{multline*}
 C\bigg(\mbf E^{x,x}_t\Big[\|L_t(Z)\|_\ga^{4\theta}\Big]^{1/2\theta}\mbf E^{x,x}_t\left[\mr e^{\theta\|L_t(Z)\|_\ga^2}\right]^{1/2\theta}
\mbf E^{y,y}_t\left[\mr e^{(\theta/2)\|L_t(Z)\|_\ga^2}\right]^{1/\theta} \\
+\mbf E^{y,y}_t\Big[\|L_t(Z)\|_\ga^{4\theta}\Big]^{1/2\theta}\mbf E^{y,y}_t\left[\mr e^{\theta\|L_t( Z)\|_\ga^2}\right]^{1/2\theta}
\mbf E^{x,x}_t\left[\mr e^{(\theta/2)\|L_t(Z)\|_\ga^2}\right]^{1/\theta}\bigg).
\end{multline*}
At this point, thanks to \eqref{Equation: C Bound}, the proof of Lemma \ref{Lemma: D Bound 2}
will be complete if we show that
\begin{align}
\label{Equation: d Exponent Bridge}
\limsup_{t\to0}t^{-\mf d}\,\sup_{x\in I}\left(\mbf E^{x,x}_t\Big[\|L_t(Z)\|_\ga^{2\theta}\Big]\right)^{1/\theta}\leq C.
\end{align}

We claim that \eqref{Equation: d Exponent Bridge} is a consequence of \eqref{Equation: d Exponent}.
To see this, we once again condition on the midpoint of $Z^{x,x}_t$:
With $\mf s(Z)<\infty$ as in \eqref{Equation: Midpoint Transition Bound}, we obtain that for any $t\in(0,1]$,
\begin{align*}
&\mbf E^{x,x}_t\Big[\|L_t(Z)\|_\ga^{2\theta}\Big]\\
&=\int_I\mbf E^{x,x}_t\Big[\|L_t(Z)\|_\ga^{2\theta}\Big|Z^{x,x}_t(t/2)=z\Big]\frac{\Pi_Z(t/2;x,z)\Pi_Z(t/2;z,x)}{\Pi_Z(t;x,x)}\d z\\
&\leq\mf s(Z)\int_I\mbf E^{x,x}_t\Big[\|L_{t/2}(Z)+L_{[t/2,t]}(Z)\|_\ga^{2\theta}\Big|Z^{x,x}_t(t/2)=z\Big]\Pi_Z(t/2;x,z)\d z\\
&\leq C\mf s(Z)\int_I\mbf E^{x,z}_t\Big[\|L_{t/2}(Z)\|_\ga^{2\theta}\Big]\Pi_Z(t/2;x,z)\d z\\
&= C\,\mbf E^x\Big[\|L_{t/2}(Z)\|_\ga^{2\theta}\Big],
\end{align*}
where the equality in the second line follows from the Doob h-transform (see \eqref{Equation: Tower+Doob}), the inequality in the fourth line follows from first applying Minkowski's inequality to bound $\|L_{t/2}(Z)+L_{[t/2,t]}(Z)\|_\ga^{2\theta}$ by $C(\|L_{t/2}(Z)\|_{\ga}^{2\theta}+\|L_{[t/2,t]}(Z)\|_{\ga}^{2\theta})$, and then using the fact that, under the conditioning $Z^{x,x}_t(t/2)=z$,
the local time processes $L_{t/2}(Z^{x,x}_t)$ and $L_{[t/2,t]}(Z^{x,x}_t)$ are i.i.d. copies of $L_{t/2}(Z^{x,z}_{t/2})$.
(We refer back to the passage following \eqref{Equation: Tower+Doob} for details.)

\subsection{Final Estimates: Proof of Lemma \ref{Lemma: Final Estimates}}
\label{Section: Variance Bounds Final}

\subsubsection{Proof of \eqref{Equation: Final Estimate General}}

We begin by proving \eqref{Equation: Final Estimate General}
in \textbf{Case 1}.
By coupling $B^{x,x}_t:=x+B^{0,0}_t$
and $\bar B^{y,y}_t:=y+\bar B^{0,0}_t$,
it follows from \eqref{Equation: Polynomial Potential Condition} that
\begin{align}\label{Equation: A Ineq}
\mc A_t(x,y)\leq2\nu t-\ka^\mf a\int_0^t\Big(\big|x+B^{0,0}_t(s)\big|^\mf a+\big|y+\bar B^{0,0}_t(s)\big|^\mf a\Big)\d s.
\end{align}
By the change of variables $s\mapsto st$ and a Brownian scaling, we then obtain
\begin{align}
\nonumber\text{r.h.s. of \eqref{Equation: A Ineq}}
&= 2\nu t-\ka^\mf a\int_0^1\Big(\big|t^{\frac{1}{\mf a}}x+t^{\frac{1}{\mf a}}B^{0,0}_t(st)\big|^\mf a+\big|t^{\frac{1}{\mf a}}y+t^{\frac{1}{\mf a}}\bar B^{0,0}_t(st)\big|^\mf a\Big)\d s\vspace{-0.1cm}\\
\nonumber &\deq2\nu t-\ka^\mf a\int_0^1\Big(\big|t^{\frac{1}{\mf a}}x+t^{\frac{1}{2}+\frac{1}{\mf a}}B^{0,0}_1(s)\big|^\mf a+\big|t^{\frac{1}{\mf a}}y+t^{\frac{1}{2}+\frac{1}{\mf a}}\bar B^{0,0}_1(s)\big|^\mf a\Big) \d s.
\end{align}

Let us introduce the shorthands
\begin{align}
\label{Equation: Compactly Supported Shorthand 0}
\ms B_{t,x}(s):=\big|t^{\frac{1}{\mf a}}x+t^{\frac{1}{2}+\frac{1}{\mf a}}B^{0,0}_1(s)\big|^\mf a,
\qquad \bar{\ms B}_{t,y}(s):= \big|t^{\frac{1}{\mf a}}y+t^{\frac{1}{2}+\frac{1}{\mf a}}\bar B^{0,0}_1(s)\big|^\mf a
\end{align}
so that, by \eqref{Equation: A Ineq}, one has
\begin{align}
\nonumber
\int_{\mbb R^2}\mbf E\left[\mr e^{4\mc A_t(x,y)}\right]^{\frac14}\d x\dd y
&\leq C\mr e^{2\nu t}
\int_{\mbb R^2}\mbf E\Big[\mr e^{-4\ka^\mf a\int_0^1\big(\ms B_{t,x}(s)+\bar{\ms B}_{t,y}(s)\big)\d s}\Big]^{\frac14} \d x\dd y\\
\label{Equaton: Variance Formula Case 1 - 1}
&=C\mr e^{2\nu t}t^{-2/\mf a}
\int_{\mbb R^2}\mbf E\Big[\mr e^{-4\ka^\mf a\int_0^1\big(\ms B_{t,t^{-1/\mf a}x}(s)+\bar{\ms B}_{t,t^{-1/\mf a}y}(s)\big)\d s}\Big]^{\frac14} \d x \dd y,
\end{align}
where in the second line we applied the change of variables $(x,y)\mapsto t^{-1/\mf a}(x,y)$.
To alleviate notation, let us henceforth write
\begin{align}
\label{Equation: Compactly Supported Shorthand}
\mc F_t(x,y):=\mr e^{-\ka^\mf a\int_0^1\big(\ms B_{t,t^{-1/\mf a}x}(s)+\bar{\ms B}_{t,t^{-1/\mf a}y}(s)\big)\d s},
\end{align}
noting that the dependence of $\mf a$ and $\ka$ are implicit in this notation. For every fixed $x,y\in\mbb R$, 
\begin{align}
\label{Equation: Compactly Supported Shorthand Limit}
\lim_{t\to0}\mc F_t(x,y)=\mr e^{-|\ka x|^\mf a-|\ka y|^\mf a}
\end{align}
almost surely. Moreover, for every $z,\bar z\in\mbb R$,
\[|z+\bar z|^{\mf a}\geq|z+\bar z|^{\min\{\mf a,1\}}-1\geq|z|^{\min\{\mf a,1\}}-|\bar z|^{\min\{\mf a,1\}}-1,\]
and therefore
\begin{align}
\label{Equation: min(a,1) Triangle Inequality}
\sup_{t\in(0,1]}&\mc F_t(x,y)^4
\leq \mr \exp\Big(-4|\ka x|^{\min\{\mf a,1\}}-4|\ka y|^{\min\{\mf a,1\}}\Big)\\
\nonumber &\times\exp\Big(4\ka^{\min\{\mf a,1\}}\big(2+\sup_{s\in[0,1]}|B^{0,0}_1(s)|^{\min\{\mf a,1\}}+\sup_{s\in[0,1]}|\bar B^{0,0}_1(s)|^{\min\{\mf a,1\}}\big)\Big).
\end{align}
We recall that the process $s\mapsto|B^{0,0}_1(s)|$ is a Bessel bridge of dimension one (e.g., \cite[Chapter XI]{RevuzYor}).
Thanks to the tail asymptotic in \cite[Remark 3.1]{GruetShi} (the Bessel bridge is denoted by $\rho$ in that paper),
we know that Bessel bridge maxima have finite exponential moments of all orders. Therefore,
since the function $\exp(-|\ka x|^{\min\{\mf a,1\}}-|\ka y|^{\min\{\mf a,1\}})$ is integrable on $\mbb R^2$,
 it follows from the
dominated convergence theorem that
\begin{align}
\label{Equation: Step 4 Case 1 General Final}
\lim_{t\to0}\int_{\mbb R^2}\mbf E[\mc F_t(x,y)^4]^{1/4}\d x\dd y
=\int_{\mbb R^2}\mr e^{-|\ka x|^\mf a-|\ka y|^\mf a}\d x \dd y=\left(\frac{2\Ga(1+1/\mf a)}{\ka}\right)^2=\frac{C_\mf a}{\ka^2}.
\end{align}
Combining \eqref{Equaton: Variance Formula Case 1 - 1}--\eqref{Equation: Step 4 Case 1 General Final}
then yields \eqref{Equation: Final Estimate General} in \textbf{Case 1}.

We now conclude the proof of \eqref{Equation: Final Estimate General}
by showing that the inequality holds also in \textbf{Case 2}.
Since $V(x)\geq|\ka x|^\mf a-\nu$,
\[\mbf E\left[\mr e^{4\mc A_t(x,y)}\right]^{1/4}\leq\mr e^{2\nu t}
\mbf E\left[\mr e^{-4\ka^\mf a\int_0^t\big(\big|X^{x,x}_t(s)\big|^\mf a+\big|\bar X^{y,y}_t(s)\big|^\mf a\big)\d s}\right]^{1/4}.\]
An application of \eqref{Equation: X and B Coupling inequality} then yields
\[\mbf E\left[\mr e^{4\mc A_t(x,y)}\right]^{1/4}\leq 2\mr e^{2\nu t}
\mbf E\left[\mr e^{-4\ka^\mf a\int_0^t\big(\big|B^{x,x}_t(s)\big|^\mf a+\big|\bar B^{y,y}_t(s)\big|^\mf a\big)\d s}\right]^{1/4};\]
hence the proof of \eqref{Equation: Final Estimate General} in \textbf{Case 2}
follows from the same argument used in {\bf Case 1}.

\subsubsection{Proof of \eqref{Equation: Final Estimate Compact 1}}

We recall that \eqref{Equation: Final Estimate Compact 1}
is in the setting of {\bf Case 1}.
By controlling $\mc A_t$ in the same way as \eqref{Equation: A Ineq},
we obtain the bound
\begin{multline}
\label{Equaton: Variance Formula Case 1 - 3 PRE}
\int_{\mbb R^2}\mbf E\left[\mr e^{4\mc A_t(x,y)}\right]^{1/4}\,\mr e^{-\frac{(|x-y|-K)^2}{2ct}}\d x\dd y\\
\leq\mr e^{2\nu t}\int_{\mbb R^2}\mbf E\Big[\mr e^{-4\ka^\mf a\int_0^1(\ms B_{t,x}(s)+ \bar{\ms B}_{t,y}(s))\d s}\Big]^{1/4}\mr e^{-\frac{(|x-y|-K)^2}{2ct}}\d x\dd y,
\end{multline}
where we recall that $\ms B_{t,x}$ and $\bar{\ms B}_{t,y}$ are denoted as
\eqref{Equation: Compactly Supported Shorthand 0}.
By the change of variables $(x,y)\mapsto t^{-1/\mf a}(x,y)$, the integral
on the right-hand side of \eqref{Equaton: Variance Formula Case 1 - 3 PRE}
is bounded above by
\begin{multline}
\label{Equaton: Variance Formula Case 1 - 3}
t^{-2/\mf a}\int_{\mbb R^2}\mbf E[\mc F_t(x,y)^4]^{1/4}\mr e^{-(|x-y|-t^{1/\mf a}K)^2/2ct^{1+2/\mf a}}\d x\dd y\\
=\sqrt{2\pi c}\cdot t^{1/2-1/\mf a}
\int_{\mbb R^2}\mbf E[\mc F_t(x,y)^4]^{1/4}\frac{\mr e^{-(|x-y|-t^{1/\mf a}K)^2/2ct^{1+2/\mf a}}}{\sqrt{2\pi c t^{1+2/\mf a}}}\d x\dd y,
\end{multline}
where we recall that $\mc F_t$ is defined as in \eqref{Equation: Compactly Supported Shorthand}.
Owing to the inequality
\[(|x-y|-t^{1/\mf a}K)^2\geq \min\{(x-y-t^{1/\mf a}K)^2, (x-y+t^{1/\mf a}K)^2\},\]
we have 
\[\mr e^{-(|x-y|-t^{1/\mf a}K)^2/2ct^{1+2/\mf a}}\leq \mr e^{-(x-y-t^{1/\mf a}K)^2/2ct^{1+2/\mf a}} + \mr e^{-(x-y+t^{1/\mf a}K)^2/2ct^{1+2/\mf a}}\]
which yields 
\[\frac{\mr e^{-(|x-y|-t^{1/\mf a}K)^2/2ct^{1+2/\mf a}}}{\sqrt{2\pi c t^{1+2/\mf a}}}
\leq\ms G_{ct^{1+2/\mf a}}(x-y-t^{1/\mf a}K)
+\ms G_{ct^{1+2/\mf a}}(x-y+t^{1/\mf a}K),\]
where we recall that $\ms G_t$ denotes the Gaussian kernel \eqref{Equation: Gaussian Kernel}.
Combining this with \eqref{Equation: min(a,1) Triangle Inequality} and substituting into \eqref{Equaton: Variance Formula Case 1 - 3} then shows that
\begin{multline}
\label{Equation: Case 1 Improved Compactly Supported Bound}
\int_{\mbb R^2}\mbf E\Big[\mr e^{-4\ka^\mf a\int_0^1(\ms B_{t,x}(s)+ \bar{\ms B}_{t,y}(s))\d s}\Big]^{1/4}\mr e^{-\frac{(|x-y|-K)^2}{2ct}}\d x\dd y\\
\leq C_\mf at^{1/2-1/\mf a}
\bigg(\int_{\mbb R^2}\mr e^{-|\ka x|^{\min\{\mf a,1\}}-|\ka y|^{\min\{\mf a,1\}}}\ms G_{ct^{1+2/\mf a}}(x-y-t^{1/\mf a}K)\d x\dd y\\
+\int_{\mbb R^2}\mr e^{-|\ka x|^{\min\{\mf a,1\}}-|\ka y|^{\min\{\mf a,1\}}}\ms G_{ct^{1+2/\mf a}}(x-y+t^{1/\mf a}K)\d x\dd y\bigg).
\end{multline}
Owing to a change of variables and the fact that the Gaussian kernel is an approximate identity,
the integrals in the right-hand side of \eqref{Equation: Case 1 Improved Compactly Supported Bound}
have the following limits by dominated convergence:
\begin{multline*}
\lim_{t\to0}\int_\mbb R\mr e^{-|\ka(x\pm t^{1/\mf a}K)|^{\min\{\mf a,1\}}}\left(\int_\mbb R\mr e^{-|\ka y|^{\min\{\mf a,1\}}} \ms G_{ct^{1+2/\mf a}}(x-y)\d y\right)\d x\\
=\int_\mbb R\mr e^{-2|\ka x|^{\min\{\mf a,1\}}}\d x=\frac{2^{1-1/\min\{\mf a,1\}} \Gamma \left(1+1/\min\{\mf a,1\}\right)}{\ka}=\frac{C_\mf a}{\ka}.
\end{multline*}
Combining this last result with \eqref{Equaton: Variance Formula Case 1 - 3 PRE} and
\eqref{Equation: Case 1 Improved Compactly Supported Bound}
concludes the proof of \eqref{Equation: Final Estimate Compact 1}.

\subsubsection{Proof of \eqref{Equation: Final Estimate Compact 2}}

We now conclude the proof of Lemma \ref{Lemma: Final Estimates}
by establishing the estimate \eqref{Equation: Final Estimate Compact 2},
which we recall is in the setting of {\bf Case 2}.
To prove this, we simply note that for any function $F$ and $x>0$, we have that
\begin{align*}
\int_0^\infty F(x,y)\Big(\mr e^{-\frac{(|x-y|-K)^2}{2ct}}+\mr e^{-\frac{(|x+y|-K)^2}{2ct}}\Big)\d y
=\int_\mbb R F(x,|y|)\mr e^{-\frac{(|x-y|-K)^2}{2ct}}\d y,
\end{align*}
and thus \eqref{Equation: Final Estimate Compact 2} is an immediate
consequence of \eqref{Equation: Final Estimate Compact 1}.
With Lemma \ref{Lemma: Final Estimates} established,
along with Lemmas \ref{Lemma: Variance Formula}--\ref{Lemma: D Bound 2},
the proof of Theorem \ref{Theorem: Variance Bound} is now fully complete.

\section{Airy-$2$ Process Counterexample}
\label{Section: Optimality}

In this section, we prove Proposition \ref{Proposition: Optimality 2}.
For every $\be>0$, let $\xi_\be$ be a Gaussian white noise with variance $1/\be$,
and define the operator
\[\hat{\mc H}^{(\be)}_{(0,\infty)}:=-\tfrac12\De+\tfrac x2+\xi_\be,\]
with a Dirichlet boundary condition at zero.
The RSO $2\hat{\mc H}^{(\be)}_{(0,\infty)}$ is widely known in the literature as the
{\bf Stochastic Airy Operator} (e.g., \cite{EdelmanSutton,RamirezRiderVirag}), and
we recall that for every $\be>0$, the {\bf Airy-${\bs\be}$ point process}, which we denote by $\mf{Ai}_\be$,
is defined as the eigenvalue point process of $-2\hat{\mc H}^{(\be)}_{(0,\infty)}$.

When $\be=2$, the Airy-$\be$ process has an alternative integrable
interpretation, namely, $\mf{Ai}_2$ is the determinantal point process
induced by the {\bf Airy kernel}
\begin{align}
\label{Eq:Airy Kernel}
\mf K(x,y):=
\begin{cases}
\displaystyle
\frac{\mr{Ai}(x)\mr{Ai}'(y)-\mr{Ai}(y)\mr{Ai}'(x)}{x-y}&\text{if }x\neq y
\vspace{5pt}\\
\mr{Ai}'(x)^2-x\mr{Ai}(x)^2&\text{if }x=y,
\end{cases}
\end{align}
where $\mr{Ai}$ denotes the Airy function
\[\mr{Ai}(x):=\frac1\pi\int_0^\infty\cos\left(\frac{u^3}{3}+xu\right)\d u,\qquad x\in\mbb R.\]

Let us denote $f_t(x):=\mr e^{tx}$ for every $t>0$.
By standard formulas for the variance of linear statistics of Determinantal point processes (e.g., \cite[Equation (8)]{G15}),
we have that\footnote{We note that the variance formula in question is typically only stated for compactly supported functions.
The result can easily be improved to \eqref{Equation: Airy Exponential Variance Formula} by using dominated convergence
with standard asymptotics for the Airy function such as \cite[10.4.59--10.4.62]{AbramowitzStegun}.}
\begin{align}\label{Equation: Airy Exponential Variance Formula}
\mbf{Var}\big[\mr{Tr}[\mr e^{-2t\hat{\mc H}^{(2)}_{(0,\infty)}}]\big]=\mbf{Var}[\mf{Ai}_2(f_t)]=\frac12\int_{\mbb R^2}\big(\mr e^{tx}-\mr e^{ty}\big)^2\,\mf K(x,y)^2\d x\dd y.
\end{align}
By expanding the square and using the identity $\mf K(x,x)=\int_{\mbb R^2}\mf K(x,y)^2\d y$
(since $\mf K$ is a symmetric projection kernel \cite[Lemma 2]{TracyWidom}), we can reformulate
this to
\[\mbf{Var}[\mf{Ai}_2(f_t)]=\int_{\mbb R}\mr e^{2tx}\,\mf K(x,x)\d x-\int_{\mbb R^2}\mr e^{t(x+y)}\,\mf K(x,y)^2\d x.\]

The computation that follows is essentially taken from
\cite{Okounkov}. We provide the full details for the reader's convenience.
Rewrite the Airy kernel as 
\[\mf K(x,y)=\int_{0}^{\infty} \mr{Ai}(u+x)\mr{Ai}(u+y) \d u\]
Then, using Fubini's theorem, we can write \eqref{Equation: Airy Exponential Variance Formula} as the difference $E_1(t)-E_2(t)$,
where 
\[E_1(t):=\int_{\mathbb{R}}\mr e^{2tx}\left(\int_{0}^{\infty} \mathrm{Ai}(u+x)^2 \d u\right)\d x\\
= \int_{0}^{\infty}\left(\int_{\mathbb{R}} \mr e^{2tx} \mathrm{Ai}(u+x)^2 \d x\right) \d u,\]
and 
\begin{align*}
E_2(t)&:=\int_{\mathbb{R}^2}\mr e^{t(x+y)}\left(\int_{0}^{\infty}\int_{0}^{\infty} \mathrm{Ai}(u+x)\mathrm{Ai}(u+y)\mathrm{Ai}(v+x)\mathrm{Ai}(v+y) \d u\dd v\right)\d x\dd y\\
&= \int_{0}^{\infty}\int_{0}^{\infty}\left(\int_{\mathbb{R}} \mr e^{tx} \mathrm{Ai}(u+x)\mathrm{Ai}(v+x) \d x\right)^2 \d u\dd v.
\end{align*}

We note that the application of Fubini in $E_1(t)$ is justified since the integrand is nonnegative,
and in $E_2(t)$ it suffices to check 
\begin{align*}
\int_{0}^{\infty}\int_{0}^{\infty} \left(\int_{\mathbb{R}}\mr e^{tx}\left\vert\mathrm{Ai}(u+x)\mathrm{Ai}(v+x) \right\vert \d x\right)^2\d u\dd v<\infty.
\end{align*}
For this, we recall the formula
\begin{align}
\label{eq:Airy Laplace}
\int_{\mathbb{R}} \mr e^{tx} \mr{Ai}(x+u) \mr{Ai}(x+v) \d x= \frac{1}{2\sqrt{\pi t}} \exp\left(\frac{t^3}{12}-\frac{u+v}{2}t-\frac{(u-v)^2}{4t}\right)
\end{align}
from \cite[Lemma 2.6]{Okounkov}, and note that by Cauchy-Schwarz, we have
\begin{align*}
\int_{\mathbb{R}} \mr e^{tx}\left\vert \mathrm{Ai}(u+x)\mathrm{Ai}(v+x)\right\vert dx
&\leq \left(\int_{\mathbb{R}} \mr e^{tx} \mathrm{Ai}(u+x)^2 dx\right)^{1/2}\left(\int_{\mathbb{R}} \mr e^{tx}\mathrm{Ai}(v+x)^2 dx\right)^{1/2}\\
&=\frac{1}{2\sqrt{\pi t}}\exp\left(\frac{t^3}{12}-\frac{u+v}{2}t\right)
\end{align*}
as desired.

With $E_1(t)$ and $E_2(t)$ established,
an application of \eqref{eq:Airy Laplace} yields
\[E_1(t)
=\int_0^\infty \frac{\exp\left(\frac{2t^3}{3}-2tu\right)}{2\sqrt{2\pi t}}\d u
= \frac{\mr e^{\frac{2 t^3}{3}}}{4 \sqrt{2 \pi } t^{3/2}}\]
and
\begin{align*}
E_2(t)=\int_{0}^{\infty}\int_0^{\infty} \frac{\exp\left(\frac{t^3}{6}-(u+v)t-\frac{(u-v)^2}{2t}\right)}{4\pi t}\d u\dd v
=\frac{\mr e^{\frac{2 t^3}{3}}}{4 \sqrt{2 \pi } t^{3/2}}\left(1-\mr{erf}\left(\frac{t^{3/2}}{\sqrt{2}}\right)\right),
\end{align*}
where $\mr{erf}(z):=\frac2{\sqrt\pi}\int_0^z\mr e^{-w^2}\d w$ denotes the error function.
Thus 
\[\lim_{t\to 0}E_1(t)-E_2(t) =
\lim_{t\to 0} \frac{\mr e^{\frac{2 t^3}{3}}}{4 \sqrt{2 \pi } t^{3/2}}\mr{erf}\left(\frac{t^{3/2}}{\sqrt{2}}\right)\\
= \frac{1}{4\pi},\]
concluding the proof of Proposition \ref{Proposition: Optimality 2}.

\appendix

\section{Transition Density Bounds}\label{Appendix: Some Stoch Analysis}

\begin{proposition}
There exists constants $0<c<C$ such that for every $t\in(0,1]$,
\begin{align}\label{Equation: Transition Bounds}
ct^{-1/2}\leq\inf_{x\in I}\Pi_Z(t;x,x)
\qquad\text{and}\qquad
\sup_{(x,y)\in I^2}\Pi_Z(t;x,y)\leq Ct^{-1/2}.
\end{align}
\end{proposition}
\begin{proof}
In \textbf{Case 1}, the result follows directly from the fact that
$\Pi_B(t;x,y)\leq1/\sqrt{2\pi t}$ and
$\Pi_B(t;x,x)=1/\sqrt{2\pi t}$ for all $x,y$ and $t$.
A similar argument holds for \textbf{Case 2}.
Consider now \textbf{Case 3}.
We recall that, by definition,
\[\Pi_Y(t;x,y):=\sum_{z\in2b\mbb Z\pm y}\ms G_t(x-z)=\frac{1}{\sqrt{2\pi t}}\left(\sum_{k\in\mbb Z}\mr e^{-(x-2bk+y)^2/2t}+\mr e^{-(x-2bk-y)^2/2t}\right).\]
On the one hand, note that $t\mapsto\mr e^{-z/t}$ is increasing in $t>0$ for every $z\geq0$; hence for every $t\in(0,1]$, one has
\begin{multline*}
\sup_{(x,y)\in(0,b)^2}\left(\sum_{k\in\mbb Z}\mr e^{-(x-2bk+y)^2/2t}+\mr e^{-(x-2bk-y)^2/2t}\right)\\
\leq\sup_{(x,y)\in(0,b)^2}\left(\sum_{k\in\mbb Z}\mr e^{-(x-2bk+y)^2/2}+\mr e^{-(x-2bk-y)^2/2}\right)<\infty.
\end{multline*}
On the other hand, by isolating the $k=0$ term in $\sum_{k\in\mbb Z}\mr e^{-(2bk)^2/2t}$,
\[\inf_{x\in(0,b)}\left(\sum_{k\in\mbb Z}\mr e^{-(2x-2bk)^2/2t}+\mr e^{-(2bk)^2/2t}\right)
\geq \left(\inf_{x\in(0,b)}\sum_{k\in\mbb Z}\mr e^{-(2x-2bk)^2/2t}\right)+1\geq1,\]
concluding the proof.
\end{proof}

\bibliographystyle{plain}
\bibliography{Bibliography}
\end{document}